\RequirePackage{amsmath}
\documentclass[letterpaper,11pt,english]{article}
\usepackage[letterpaper, portrait, margin=2.5cm]{geometry}

\usepackage{amssymb, amsfonts}
\usepackage{graphicx}
\usepackage{amsthm}
\usepackage{algorithm}
\usepackage[noend]{algpseudocode}
\usepackage[nolist]{acronym}
\usepackage[capitalize,nameinlink,noabbrev]{cleveref}
\usepackage{todonotes}
\usepackage{mathtools}
\usepackage{thm-restate}
\newacro{ses}[SES]{smallest enclosing sphere}
\usepackage{csquotes}
\usepackage[inline,shortlabels]{enumitem}

\usepackage{todonotes}
\usepackage[compatibility=false]{caption}
\usepackage{subcaption}

\newtheorem{theorem}{Theorem}
\newtheorem{lemma}[theorem]{Lemma}

\newtheorem{definition}{Definition}

\def\Fsync/{$\mathcal{F}$\textsc{sync}}
\def\Gathering/{\textsc{Gathering}}
\def\chainForm/{\textsc{Chain-Formation}}
\def\patternForm/{\textsc{Pattern Formation}}
\def\Point/{\textsc{Point}}
\def\UniformCircle/{\textsc{Uniform Circle}}
\def\gtc/{\textsc{Go-To-The-Center}}
\def\gtcShort/{\textsc{GTC}}
\def\mobs/{\textsc{Move-on-Bisector}}
\def\gtcThreeD/{\textsc{3d-Go-To-The-Center}}
\def\gtcThreeDShort/{\textsc{3d-GTC}}
\def\gtcThreeDCont/{\textsc{Continuous-3d-Go-To-The-Center}}
\def\gtcThreeDContShort/{\textsc{Cont-3d-GTC}}
\def\tangentialNormal/{tangential-normal}
\def\moveOnAngleMinimizer/{\textsc{Move-on-Angle-Minimizer}}
\def\look/{\texttt{Look}}
\def\compute/{\texttt{Compute}}
\def\move/{\texttt{Move}}

\def\LCMlong/{\textsc{Look-Compute-Move}}
\def\LCM/{LCM}

\def\chainForm/{\textsc{Chain-Formation}}
\def\Gathering/{\textsc{Gathering}}
\def\maxForm/{\textsc{Max-Chain-Formation}}

\def\Oblot/{\ensuremath{\mathcal{OBLOT}}}
\def\Luminous/{\ensuremath{\mathcal{LUMINOUS}}}

\newcommand{\fsync}{\textsc{$\mathcal{F}$sync}}
\newcommand{\ssync}{\textsc{$\mathcal{S}$sync}}
\newcommand{\async}{\textsc{$\mathcal{A}$sync}}

\newcommand{\viewingRadiusConstant}{\ensuremath{4}}

\newcommand{\runState}{run}
\newcommand{\runInit}{Init-\-Sta\-te}
\newcommand{\jointRunInit}{\text{Joint Init-State}}

\newcommand{\hop}{\text{hop}}
\newcommand{\jointHop}{\text{joint hop}}
\newcommand{\Hop}{\text{Hop}}
\newcommand{\JointHop}{\text{Joint Hop}}

\newcommand{\shorten}{\text{shorten}}
\newcommand{\jointShorten}{\text{joint shorten}}
\newcommand{\Shorten}{\text{Shorten}}
\newcommand{\JointShorten}{\text{Joint Shorten}}

\newcommand{\merge}{\text{merge}}
\newcommand{\jointMerge}{\text{joint merge}}
\newcommand{\Merge}{\text{Merge}}
\newcommand{\JointMerge}{\text{Joint Merge}}

\newcommand{\schlaefli}{\text{star-\-op\-er\-a\-tion}}
\newcommand{\bisector}{bisector-op\-er\-a\-tion}

\newcommand{\Schlaefli}{\text{Star-\-Op\-er\-a\-tion}}
\newcommand{\Bisector}{Bisector-Op\-er\-a\-tion}

\newcommand{\direction}[2]{\ensuremath{r({#1}, {#2 +1})}}

\newcommand{\shortenAngle}{\ensuremath{\frac{7}{8}\pi}}

\newcommand{\bisectorConstant}{\ensuremath{\frac{1}{5}}}
\newcommand{\shortenConstant}{\ensuremath{0.019}}

\DeclareMathOperator{\sign}{sgn}

\makeatletter
\def\@parfont{\bfseries}
\makeatother

\usepackage{authblk}

\title{Gathering a Euclidean Closed Chain of Robots in Linear Time
\thanks{This paper is a full version of the brief announcement presented at SSS 2020.
}}

\author[]{Jannik Castenow}
\author[]{Jonas Harbig}
\author[]{Daniel Jung}
\author[]{Till Knollmann}
\author[]{Friedhelm Meyer auf der Heide}
\affil[]{Heinz Nixdorf Institute and Department of Computer Science\\
	Paderborn University,
 \{janniksu, jharbig, jungd, tillk, fmadh\}@mail.upb.de}

\date{}

\begin{document}

\maketitle

\begin{abstract}
	This work focuses on the following question related to the \Gathering/ problem of $n$ autonomous, mobile robots in the Euclidean plane:
	Is it possible to solve \Gathering/ of robots that do not agree on any axis of their coordinate systems (disoriented robots) and see other robots only up to a constant distance (limited visibility)
	in $o(n^2)$ fully synchronous rounds (the \fsync{} scheduler)?
	The best known algorithm that solves \Gathering/ of disoriented robots with limited visibility  in the \Oblot/ model (oblivious robots) needs $\Theta\left(n^2\right)$ rounds \cite{DBLP:conf/spaa/DegenerKLHPW11}.
	The lower bound for this algorithm even holds in a simplified closed chain model, where each robot has exactly two neighbors and the chain connections form a cycle.
   The only existing algorithms achieving a linear number of rounds for disoriented robots assume robots that are located on a two dimensional grid \cite{gridChain} and \cite{gridOptimal}.
	Both algorithms make use of locally visible lights to communicate state information (the \Luminous/ model).

	In this work, we show for the closed chain model, that $n$ disoriented robots with limited visibility in the Euclidean plane can be gathered in $\Theta\left(n\right)$ rounds assuming the \Luminous/ model.
	The lights are used to initiate and perform so-called runs along the chain.
	For the start of such runs, locally unique robots need to be determined.
	In contrast to the grid \cite{gridChain}, this is not possible in every configuration in the Euclidean plane.
	Based on the theory of isogonal polygons by Branko Gr\"unbaum, we identify the class of  isogonal configurations in which -- due to a high symmetry -- no such locally unique robots can be identified.
	Our solution combines two algorithms: The first one gathers isogonal configurations; it works without any lights.
	The second one works for non-isogonal configurations; it identifies locally unique robots to start runs, using a constant number of lights.
	Interleaving these algorithms solves the \Gathering/ problem in $\mathcal{O}\left(n\right)$ rounds.
\end{abstract}


\section{Introduction}
The \Gathering/ problem is one of the most studied and fundamental problems in the area of distributed computing by mobile robots.
\Gathering/ requires a set of initially scattered point-shaped robots to meet at the same (not predefined) position.
This problem has been studied under several different robot and time models all having in common that the capabilities of the individual robots are very restricted.
The central questions among all these models are: Which capabilities of robots are needed to solve the \Gathering/ problem and how do these capabilities influence the runtime?
While the question about solvability is quite well understood nowadays, much less is known concerning the question how the capabilities influence the runtime.
The best known algorithm -- \textsc{Go-To-The-Center} (see \cite{localgathering} and \cite{DBLP:conf/spaa/DegenerKLHPW11} for the runtime analysis) -- for $n$ disoriented robots (no agreement on the coordinate systems) in the Euclidean plane with limited visibility in the \Oblot/ model, assuming the \fsync{} scheduler (robots operate in fully synchronized Look-Compute-Move cycles), requires $\Theta(n^2)$ rounds.
The fundamental features of the \Oblot/ model are that the robots are \emph{autonomous} (they are not controlled by a central instance), \emph{identical} and \emph{anonymous} (all robots are externally identical and do not have unique identifiers), \emph{homogeneous} (all robots execute the same algorithm), \emph{silent} (robots do not communicate directly) and  \emph{oblivious} (the robots do not have any memory of the past).

The best lower bound for \Gathering/ $n$ disoriented robots with limited visibility in the \Oblot/ model, assuming the \fsync{} scheduler, is the trivial  $\Omega(n)$ bound.
Thus, more concretely the above mentioned question can be formulated as follows:
Is it possible to gather $n$ disoriented robots with limited visibility in the Euclidean plane in $o(n^2)$ rounds, and which capabilities do the robots need?

The $\Omega(n^2)$ lower bound for \textsc{Go-to-the-Center} examines an initial configuration where the robots form a cycle with neighboring robots having a constant distance, the \emph{viewing  radius}.
It is shown that  \textsc{Go-to-the-Center} takes $\Omega(n^2)$ rounds until the robots start to see more robots than their initial neighbors.
Thus, the lower bound holds even in a simpler closed chain model, where the robots form an arbitrarily winding closed chain, and each robot sees exactly its  two direct neighbors.

Our main result is that this quadratic lower bound can be beaten for the closed chain model, if we extend the \Oblot/ model by allowing each robot a constant number of visible lights, which can be seen by the neighboring robots, i.e. if we allow \Luminous/ robots, compare \cite{luminous}.
For this model we present an algorithm with linear runtime.
In the algorithm, the lights are used to initiate and perform so-called runs along the chain.
For the start of such runs, locally unique robots need to be determined.
In contrast to the grid \cite{gridChain}, this is not possible in every configuration in the Euclidean plane.
Based on the theory of isogonal polygons by Branko Grünbaum \cite{grunbaum1994metamorphoses}, we identify the class of  isogonal configurations in which -- due to a high symmetry -- no such locally unique robots can be identified.
Our solution combines two algorithms: The first one gathers isogonal configurations; it works without any lights.
The second one works for non-isogonal configurations; it identifies locally unique robots to start runs, using a constant number of lights.
Interleaving these algorithms solves the \Gathering/ problem in $\mathcal{O}\left(n\right)$ rounds.

\paragraph{Related Work}

A lot of research is devoted to \Gathering/ in less synchronized settings (\async{} and \ssync{}) mostly combined with an unbounded viewing radius.
Due to space constraints, we omit the discussion here, for more details see e.g.\ \cite{localgathering,DBLP:journals/siamcomp/CieliebakFPS12,DBLP:conf/sirocco/CieliebakP02,DBLP:journals/siamcomp/CohenP05,DBLP:journals/tcs/FlocchiniPSW05,DBLP:journals/tcs/Prencipe07,DBLP:journals/siamcomp/SuzukiY99}.
Instead, we focus on results about the synchronous setting (\fsync{}), where algorithms as well as runtime bounds are known.
For a comprehensive overview over models, algorithms and analyses, we refer the reader to the recent survey \cite{DBLP:series/lncs/Flocchini19}.
In the \Oblot/ model, there is the \textsc{Go-To-The-Center} algorithm\ \cite{localgathering} that solves \Gathering/ of disoriented robots with local visibility in $\Theta\left(n^2\right)$ rounds assuming  the \fsync{} scheduler \cite{DBLP:conf/spaa/DegenerKLHPW11}.
The same runtime can be achieved for robots located on a two dimensional grid \cite{gridNoLights}.
It is conjectured that both algorithms are asymptotically optimal and thus $\Omega \left(n^2\right)$ is also a lower bound for \emph{any} algorithm that solves \Gathering/ in this model.
Interestingly, a lower bound of $\Omega(D_G^2)$ has been shown for any \emph{conservative} algorithm, i.e.\ an algorithm that only increments the edge set of the visibility graph.
Here, $D_G \in \Theta \left(\sqrt{\log n}\right)$ denotes the diameter of the initial visibility graph
\cite{DBLP:journals/tcs/IzumiKPT18}.
On the first sight, being conservative seems to be a significant restriction.
However, all known algorithms solving \Gathering/ with limited visibility are indeed conservative.
It is open whether this lower bound can be extended to diameters of larger size.

Faster runtimes could so far only be achieved by assuming agreement on one or two axes of the local coordinate systems or considering the \Luminous/ model.
In \cite{DBLP:conf/sss/PoudelS17}, an universally optimal algorithm with runtime $\Theta\left(D_E\right)$ for robots in the Euclidean plane assuming one-axis agreement in the \Oblot/ model is introduced.
$D_E$ denotes the Euclidean diameter (the largest distance between any pair of robots) of the initial configuration.
Beyond the one-axis agreement, their algorithm crucially depends on the distinction between the viewing range of a robot and its \emph{connectivity range}:
Robots only consider other robots within their connectivity range as their neighbors but can see farther beyond.
Notably, this algorithms also works under the \async{} scheduler.

Assuming disoriented robots, the algorithms that achieve a runtime of $o(n^2)$ are developed under the \Luminous/ model and assume robots that are located on a two dimensional grid:
There exist two algorithms having an asymptotically optimal runtime of $\mathcal{O}\left(n\right)$; one algorithm for \emph{closed chains} \cite{gridChain} and another one for arbitrary (connected) swarms \cite{gridOptimal}.
Following the notion of \cite{gridChain}, we consider a closed chain of robots in this work.
In a closed chain, the robots form a winding, possibly self-intersecting, chain where the distance between two neighbors is upper bounded by the connectivity range and the robots can see a constant distance along the chain in each direction, denoted as the viewing range.
The main difference between a closed chain and arbitrary (connected) swarms is that in the closed chain, a robot only observes a constant number of its direct neighbors while in arbitrary swarms all robots in the viewing range of a robot are considered.
Interestingly, the lower bound of the \textsc{Go-To-The-Center} algorithm \cite{DBLP:conf/spaa/DegenerKLHPW11}, holds also for the closed chain model.

\paragraph{Our Contribution}
In this work, we give the first asymptotically optimal algorithm that solves \Gathering/ of disoriented robots in the Euclidean plane.
More precisely, we show that a closed chain of disoriented robots with limited visibility located in the Euclidean plane can be gathered in $\mathcal{O} \left(n\right)$ rounds assuming the \Luminous/  model with a constant number of lights and the \fsync{} scheduler.
This is asymptotically optimal, since if the initial configuration forms a straight line at least $\Omega(n)$ rounds are required.

\begin{theorem}
	For any initially connected closed chain of disoriented robots in the Euclidean plane with a viewing range of $4$ and a connectivity range of $1$, \textnormal{\Gathering/} can be solved in $\mathcal{O}\left(n\right)$ rounds assuming the \textnormal{\fsync{} } scheduler and a constant number of visible lights.
	The number of rounds is asymptotically optimal.
\end{theorem}

The visible lights help to exploit asymmetries in the chain to identify locally unique robots that generate \emph{runs}.
One of the major challenges is the handling of highly symmetric configurations.
While it is possible to identify locally unique robots in every connected configuration on the grid (as it is done in the algorithm for closed chains on the grid \cite{gridChain}), this is \emph{impossible} in the Euclidean plane.
We identify the class of \emph{isogonal configurations} based on the theory of isogonal polygons by Grünbaum \cite{grunbaum1994metamorphoses} and show that no locally unique robots can be determined in these configurations, while this is possible in every other configuration.
We believe that this characterization is of independent interest because highly symmetric configurations often cause a large runtime.
For instance, the lower bound of the \textsc{Go-To-The-Center} algorithm holds for an isogonal configuration \cite{DBLP:conf/spaa/DegenerKLHPW11}.

Our approach combines two algorithms into one: An algorithm inspired by \cite{gridChain,DBLP:journals/tcs/KutylowskiH09} that gathers non-isogonal configurations in linear time using visible lights and another algorithm for isogonal configurations without using any lights.
Note that there might be cases in which both algorithms are executed in parallel due to the limited visibility of the robots.
An additional rule ensures that both algorithms can be interleaved without hindering each other.

\section{Model \& Notation}

We consider $n$ robots $r_0, \dots, r_{n-1}$ connected in a closed chain.
Every robot $r_i$ has two direct neighbors: $r_{i-1}$ and $r_{i+1}$.\footnote{Throughout this work, all operations on indices have to be understood modulo $n$.}
The connectivity range is assumed to be $1$, i.e.\ two direct neighbors are allowed to have a distance of at most $1$.
The robots are disoriented, i.e.; they do not agree on any axis of their local coordinate systems and the latter can be arbitrarily rotated.
This also means that there is no common understanding of left and right.
However, the robots agree on unit distance and are able to measure distances precisely.
Except of their direct neighbors, robots have a constant viewing radius along the chain.
Each robot can see its \viewingRadiusConstant{} predecessors and successors along the chain.
We assume the \Luminous/ model: the robots have a constant number of locally visible states (lights) that can be perceived by all robots in their neighborhood.
Consider the robot \( r_{i} \) in the round \( t \).
Let $p_i(t) = \left(
x_i(t) , y_i(t) \right)^{T}$ be the position of $r_i$ in round \( t \) in a global coordinate system (not known to the robots) and $d(p_i(t),p_j(t)) = \|p_i(t) - p_j(t)\|$ the Euclidean distance between the robots $r_i$ and $r_j$ in round $t$.
Furthermore, let $u_i(t) = p_{i}(t) - p_{i-1}(t)$ be the vector pointing from robot $r_{i-1}$ to $r_i$ in round $t$.
The length of the chain is defined as $L(t) := \sum_{i=0}^{n-1} \|u_i(t)\|$.
$\alpha_{i}(t) \in [0,\pi]$ denotes the angle between $u_{i}(t)$ and $u_{i+1}(t)$.
$\sign_i\left(\alpha_{i}(t)\right) \in \{-1,0,1\}$ denotes the orientation of the angle $\alpha_{i}(t)$ from $r_i$'s point of view (this can differ from robot to robot) and $\sign(\alpha_{i}(t))$ denotes the orientation in a global coordinate system.
$N_i(t) = \{r_{i-3}, \dots, r_{i+3}\}$ is the neighborhood of $r_i$ in round \( t \).
Throughout the execution of the algorithm it can happen that two robots merge and continue to behave as a single robot.
$r_{i+1}$ always represents the first robot with an index larger $i$ that has not yet merged with $r_i$.
$r_{i-1}$ is defined analogously.

\section{Algorithm}

Our approach consists of two algorithms -- one for asymmetric configurations and one for highly symmetric (isogonal) configurations.
Since robots cannot perceive the entire chain, it might happen that both algorithms are executed in parallel, i.e.; some robots might move according to the highly symmetric algorithm while others follow the asymmetric algorithm.
In case of asymmetric configurations, the visible states are used to sequentialize the movements of robots such that no sequence of three neighboring robots moves in the same round.
The impossibility of ensuring this in highly symmetric configurations raises the need for the additional algorithm.

\subsection{High Level Description}
The main concept of the asymmetric algorithm is the notion of a \emph{\runState}.
A \runState{} is a visible state (realized with lights) that is passed along the chain in a fixed direction associated to it.
Robots with a \runState{} perform a movement operation while robots without do not.
The movement is sequentialized in a way that in round $t$ the robot $r_i$ executes a move operation (and neither $r_{i-1}$ nor $r_{i+1}$), the robot $r_{i+1}$ in round $t+1$ and so on.
The movement of a \runState{} along the chain can be seen in \cref{figure:how-runs-are-passed}.

\begin{figure}[H]
	\centering
	\begin{minipage}[c]{0.5\textwidth}
		\includegraphics[page=8, width=\textwidth]{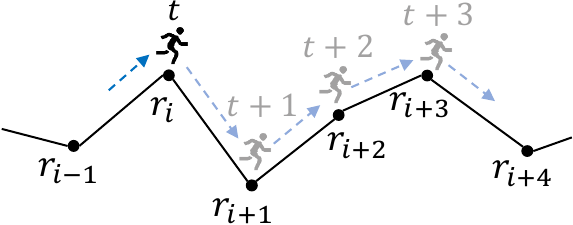}
	\end{minipage}
	\hfill
	\begin{minipage}[c]{0.45\textwidth}
		\captionof{figure}{A \runState{} located at \( r_{i} \) in round \( t \) is passed in its direction along the chain, i.e.; it is located at \( r_{i + x} \) in round \( t + x \).}
		\label{figure:how-runs-are-passed}
	\end{minipage}
\end{figure}

This way, any moving robot does not have to consider movements of its neighbors since it knows that the neighbors do not change their positions.
The only thing a robot has to ensure is that the distance to its neighbors stays less or equal to $1$.
To preserve the connectivity of the chain only two run patterns are allowed: a robot $r_i$ has a \runState{} and either none of its direct neighbors (an \emph{isolated run}) has a run or exactly one of its direct neighbors has a run such that the runs are heading in each other's direction (a \emph{joint run-pair}).
This essentially means that there are no sequences of \runState{}s of length at least $3$.

For robots with a \runState{}, there are three kinds of movement operations, the \emph{\merge{}}, the \emph{\shorten{}} and the \emph{\hop{}}.
The purpose of the \textbf{\merge{}} is to reduce the number of robots in the chain.
It is executed by a robot $r_i$  if its neighbors have a distance of at most $1$.
In this case, $r_i$ is not necessary for the connectivity of the chain and can safely be removed.
Removing \( r_{i} \) means, that it jumps onto the position of its next neighbor in the direction of the run, the robots merge their neighborhoods and both continue to behave as a single robot in future rounds.
The execution of a \merge{} stops a \runState{}.
Moreover, some additional care has to be taken here: removing robots from the chain decreases the distance of nearby runs.
In the worst-case it could happen that another run pattern besides the isolated run and the joint run-pair is executed.
To avoid such a situation, a \merge{} stops all \runState{}s that might be present in the neighborhoods of the two merging robots.
The goal of a \textbf{\shorten{}} is to reduce the length of the chain.
Intuitively, if the angle between vectors of $r_i$ pointing to its neighbors is not too large, it can reduce the length of the chain by jumping to the midpoint between its neighbors in many cases.
This is denoted as a \shorten{}.
The execution of a \shorten{} also stops a run.
In case no progress (in terms of reducing the number of robots or the length of the chain) can be made locally, a \textbf{\hop{}} is executed.
The purpose of a \hop{} is to exchange two neighboring vectors in the chain.
By this, each \runState{} is associated with a run vector that is swapped along the chain until it finds a position at which progress can be made.
For each of the three operations, there is also a \emph{joint} one (\jointHop, \jointShorten{} and \jointMerge{}) which is a similar operation executed by a joint run-pair.

The main question now is \emph{where} \runState{}s should be started.
For this, we identify robots that are -- regarding their local neighborhood -- geometrically unique.
These robots are assigned an \runInit{} allowing them to regularly generate new \runState{}s.
During the generation of new \runState{}s it is ensured that a certain distance to other \runState{}s is kept.
In isogonal configurations, however, it is not possible to identify locally unique robots.
To overcome this, we introduce an additional algorithm for these configurations.
Isogonal configurations have in common that all robots lie on the boundary of a common circle.
We exploit this fact by letting the robots move towards the center of the surrounding circle in every round until they finally gather in its center.
Additional care has to be taken in case both algorithms interfere with each other.
This can happen if some parts of the chain are isogonal while other parts are asymmetric.
Since a robot can only decide how to move based on its local view, the robots behave according to different algorithms in this case.
We show how to handle such a case and ensure that the two algorithms do not hinder each other later.

\subsection{Additional Notation}

For a robot $r_i$, $\mathit{init}(r_i,t) = \mathit{true}$ if $r_i$ has an \runInit{} and $\mathit{run}(r_i,t) = \mathit{true}$ if $r_i$ has a \runState{} in round $t$.
Additionally, $\mathit{init}(N_i(t)) = \{r_j \in N_i(t) |\, \mathit{init}(r_j,t) = \mathit{true}\}$ and $\mathit{run}(N_i(t)) = \{r_j \in N_i(t) |\, \mathit{run}(r_j,t) = \mathit{true}\}$.
Let $\kappa$ denote an arbitrary run.
$r(\kappa,t)$ denotes the robot that has run $\kappa$ in round $t$ and \direction{\kappa}{t} denotes the robot that will have run $\kappa$ in round $t+1$ (the direction of $\kappa$).
In addition, $v_{\kappa}$ denotes the run-vector of $\kappa$ and $p_\kappa(t) = p_{r(\kappa,t)}(t)$.

\subsection{Asymmetric Algorithm} \label{section:asymmetric}

The asymmetric algorithm consists of two parts: the generation of new \runState{}s and the movement depending on such a \runState{}.
We start with explaining the movement depending on \runState{}s and explain afterwards how the \runState{}s are generated in the chain.
To preserve the connectivity of the chain, it is ensured that at most two directly neighboring robots execute a move operation in the same round.
This is done by allowing the existence of only two patterns of \runState{}s at neighboring robots: either $r_i$ and neither $r_{i-1}$ nor $r_{i+1}$ has a \runState{} (isolated run) or $r_i$ and $r_{i+1}$ have \runState{}s heading in each other's direction while \( r_{i-1} \) and \( r_{i+2} \) do not have runs (joint run-pair).
All other patterns, especially sequences of length at least $3$ of neighboring robots having \runState{}s are prohibited by the algorithm.

\begin{definition}
	A run $\kappa$ is called an \emph{isolated} run in round $t$ if $r(\kappa,t) = r_i$ and $\mathit{run}(r_{i-1},t)$ $= \mathit{run}(r_{i+1},t) = \mathit{false}$.
	Two runs $\kappa_1$ and $\kappa_2$ with $r(\kappa_1,t) = r_i$ and $r(\kappa_2,t) = r_{i+1}$ are called a \emph{joint run-pair} in round $t$ in case \direction{\kappa_1}{t} = $r_{i+1}$, $\direction{\kappa_2}{t} = r_i$ and $\mathit{run}(r_{i-1},t) = \mathit{run}(r_{i+2},t) = \mathit{false}$.
\end{definition}

In the following, we describe the concrete movements for isolated \runState{}s and joint run-pairs.
The formal definitions of the operations are given afterwards.
Assume that the number of robots in the chain is at least $6$ and consider an isolated \runState{} $\kappa$ in round $t$ with $r(\kappa,t) =r_i$ and $ \direction{\kappa}{t} = r_{i+1}$.
Then, $r_i$ moves as follows: (the cases are checked with decreasing priority).
\medskip

\begin{tabular}[H]{rll}
	1. & If $d(p_{i-1}(t), p_{i+1}(t)) \leq 1$, \hspace*{0.1cm} & $r_i$ executes a \merge{}.                                   \\
	2. & If $d(p_{i}(t), p_{i+2}(t)) \leq 1$,                   & $r_i$ does not move and passes the \runState{} to $r_{i+1}$. \\
	3. & If $\alpha_i(t) \leq \shortenAngle{}$,                 & $r_i$ executes a \shorten{}.                                 \\
	4. & Otherwise,                                             & $r_i$ executes a \hop{}.
\end{tabular}
\noindent\medskip

Now consider joint run-pair at robots $r_i$ and $r_{i+1}$.
The robots $r_i$ and $r_{i+1}$ move as follows (the cases are checked with decreasing priority):
\medskip

\begin{tabular}[H]{rll}
	1. & If $d(p_{i-1}(t), p_{i+2}(t)) < 2$,                                                                        & both execute a \jointMerge{}.         \\
	2. & If $\alpha_{i}(t) \leq \shortenAngle{}$ \text{and} $\alpha_{i+1}(t) \leq \shortenAngle{}$, \hspace*{0.1cm} & both execute a \jointShorten{}        \\
	3. & If $\alpha_{i}(t) \leq \shortenAngle{}$,                                                                   & only $r_i$ executes a \shorten{}.     \\
	4. & If $\alpha_{i+1}(t) \leq \shortenAngle{}$,                                                                 & only $r_{i+1}$ executes a \shorten{}. \\
	5. & If $\angle (u_i(t), -u_{i+2}(t)) \leq \shortenAngle$.                                                      & both execute a \jointShorten{}.       \\
	6. & Otherwise,                                                                                                 & both execute a \jointHop{}.
\end{tabular}
\noindent\medskip

If the number of robots in the chain is at most $5$ (robots can detect this since they can see \viewingRadiusConstant chain neighbors in each direction), the robots move a distance of $1$ towards the center of the smallest enclosing circle of their neighborhood.
This ensures \Gathering/ after at most $5$ more rounds.
The concrete movement operations are defined as follows:

\paragraph{\Hop{} and \JointHop{}}
Consider an isolated run $\kappa$ with $r(\kappa,t) = r_i$ and the direction $\direction{\kappa}{t} = r_{i+1}$.
Assume that $r_i$ executes a \emph{\hop{}}.
The new position \( p_i(t+1) \) is \( p_{i+1}(t) - u_i(t) \).
The \runState{} continues in its direction, more precisely, in round $t+1$, $r\left(\kappa,t+1\right) = r_{i+1}$ and $r(\kappa,t+2) = r_{i+2}$.
A \emph{\jointHop{}} is a similar operation executed by two neighboring robots $r_i$ and $r_{i+1}$ with a joint run-pair.
Assume that $r(\kappa_1,t) = r_i$, $\direction{\kappa_1}{t} = r_{i+1}$, $r(\kappa_2,t) = r_{i+1}$ and $\direction{\kappa_2}{t} = r_i$.
The new positions are $p_i(t+1) = p_{i-1}(t) + u_{i+2}(t)$ and $p_{i+1}(t+1) = p_{i+2}(t) - u_i(t)$.
Both \runState{}s continue in their direction while skipping the next robot: in round $t+1$, $r\left(\kappa_1,t+1\right) = r_{i+2}$, $r(\kappa_1,t+2) = r_{i+3}$, $r\left(\kappa_2,t+1\right) = r_{i-1}$ and $r(\kappa_2,t+2) = r_{i-2}$.
See \Cref{figure:hopandjointhop} for a visualization of a \hop{} and a \jointHop{}.

\begin{figure}[htb]
	\centering
	\includegraphics[width=0.8\textwidth, page=1]{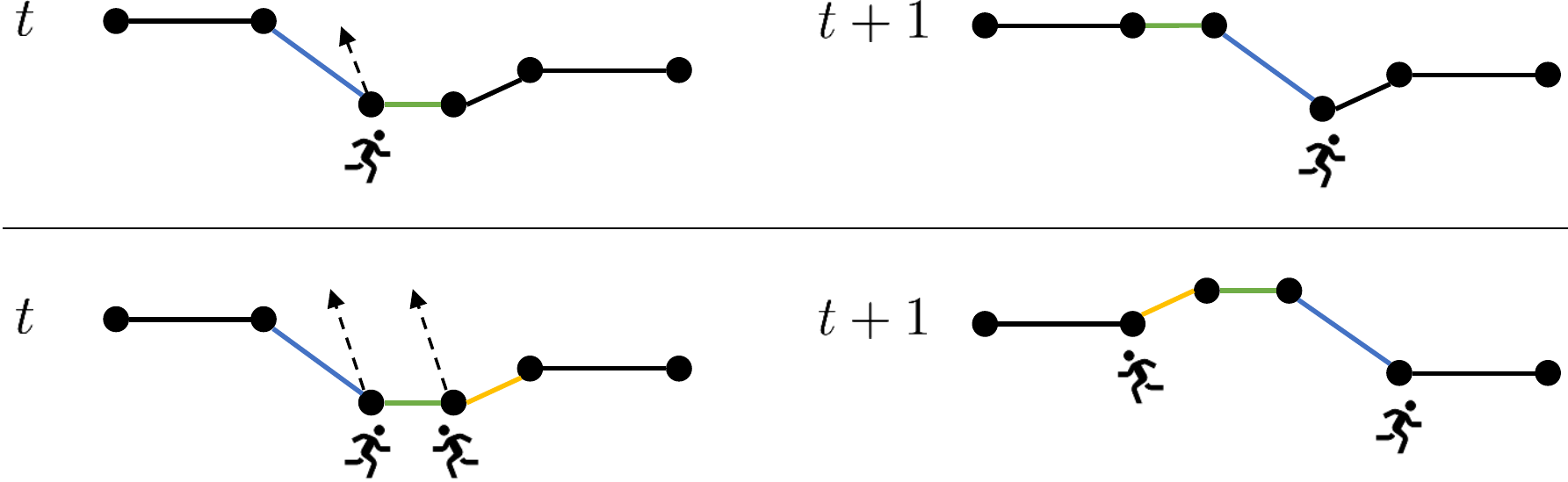}
	\caption{Visualization of a \hop{} (above) and a \jointHop{} (below).}
	\label{figure:hopandjointhop}
\end{figure}

\paragraph{\Shorten{} and \JointShorten{}}
In the \emph{\shorten{}}, a robot $r_i$ with an isolated run moves to the midpoint between its neighbors: $p_i(t+1) = \frac{1}{2} \cdot  p_{i-1}(t) + \frac{1}{2} \cdot p_{i+1}(t)$.
The run stops.
In a \emph{\jointShorten{}} executed by two robots $r_i$ and $r_{i+1}$ with a joint run-pair,
the vector $v = p_{i+2}(t) - p_{i-1}(t)$ is subdivided into three parts of equal length.
The new positions are $p_{i}(t+1) = p_{i-1}(t)+ \frac{1}{3} \cdot  v$ and $p_{i+2}(t+1) = p_{i+2}(t) - \frac{1}{3} \cdot v$.
Both runs are stopped after executing a \jointShorten{}.
See \Cref{figure:shortenandjointshorten} for a visualization of both operations.

\begin{figure}[htb]
	\centering
	\includegraphics[width=0.7\textwidth, page=3]{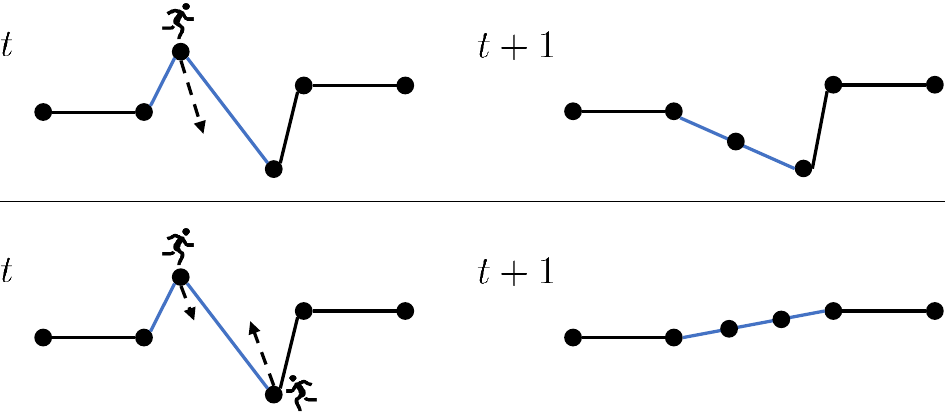}
	\caption{Visualization of a \shorten{} (above) and a \jointShorten{} (below).}
	\label{figure:shortenandjointshorten}
\end{figure}

\paragraph{\Merge{} and \JointMerge{}}
Consider an isolated run $\kappa$ with $r(\kappa,t) = r_i$ and $\direction{\kappa}{t}$ $= r_{i+1}$.
If $r_i$ executes a \emph{\merge{}}, it moves to $p_{i+1}(t)$.
The robots $r_i$ and $r_{i+1}$ merge such that  their neighborhoods are merged and they continue to behave as a single robot.
In a \emph{\jointMerge{}}, the robots $r_i$ and $r_{i+1}$ both move to $\frac{1}{2}p_{i-1}(t) + \frac{1}{2}p_{i+1}(t)$ and merge there.
Afterwards, they have identical neighborhoods and behave as a single robot.
All runs that participate in a \merge{} or a \jointMerge{} are immediately stopped.
Beyond that, all \runState{}s in the neighborhood of $r_i$ and $r_{i+1}$ are immediately stopped and all robots in $N_i(t)$ do not start any further \runState{}s within the next $4$ rounds (the robots are \emph{blocked}).
\Cref{figure:mergejointmerge} visualizes both operations.
Special care has to be taken of \runInit{}s.
Suppose that a robot $r_i$ executes a \merge{} into the direction of $r_{i+1}$ while having an \runInit{}.
The \runInit{} is handled as follows: In case $\mathit{init}(r_{i+2}) = \mathit{false}$ and $r_{i+2}$ does not execute a \merge{} in the same round and $\mathit{init}(r_{i+3}) = \mathit{true}$, the \runInit{} of $r_i$ is passed to $r_{i+1}$.
Otherwise the state is removed.

\begin{figure}[htb]
	\centering
	\includegraphics[width=0.677\textwidth, page=2]{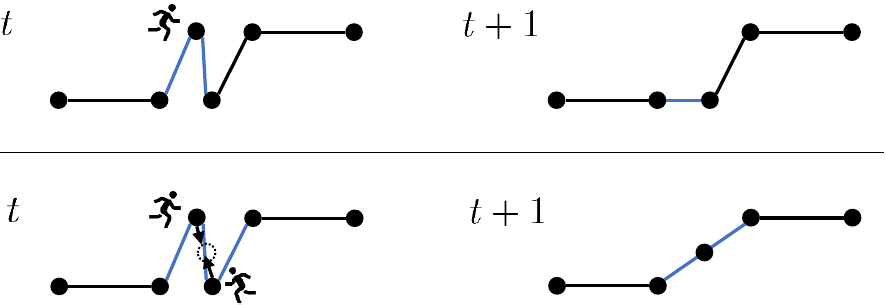}
	\caption{Visualization of a \merge{} (above) and a \jointMerge{} (below).}
	\label{figure:mergejointmerge}
\end{figure}

\paragraph{Where to start runs?}
New \runState{}s are created by robots with \runInit{}s.
To generate new \runInit{}s, we aim at discovering structures in the chain that are asymmetric.
When such a structure is observed by the surrounding robots, the robot closest to the structure is assigned a so-called \runInit{}.
Such a robot can thus remember that it was at a point of asymmetry to generate runs in the future.
To keep the distance between \runState{}s (important for maintaining the connectivity) our rules ensure that at most two neighboring robots have an \runInit{}.
Sequences of length at least $3$ of robots having an \runInit{} are prohibited.
Intuitively, there are three sources of asymmetry in the chain: sizes of angles, orientations of angles and lengths of vectors.
To detect an asymmetry, we introduce \emph{patterns} depending on the size of angles, the orientation of angles and the vector lengths.
To avoid too many fulfilled patterns, the next class of patterns is only checked if a full symmetry regarding the previous pattern is identified.
More precisely, a robot only checks orientation patterns in case all angles $\alpha_i(t)$ in its neighborhood are identical.
Similarly, a robot only checks vector length patterns if all angles in its neighborhood have the same size and the same orientation.
Whenever a pattern holds true, the robot observing the pattern assigns itself an \runInit{} if there is no other robot already assigned an \runInit{} in its neighborhood.
If it happens that two direct neighbors are assigned an \runInit{}, they fulfilled the same type of pattern and form a \jointRunInit{} together.

\paragraph{Angle Patterns}

A robot \( r_{i} \) is assigned an \runInit{} if either \( \alpha_{i-1}(t) > \alpha_{i}(t) \leq \alpha_{i+1}(t) \) or  \( \alpha_{i-1}(t) \geq \alpha_{i}(t) < \alpha_{i+1}(t) \).
Intuitively, the robot is a point of asymmetry if its angle is a local minimum.
\cref{figure:angle-patterns} shows an example.

\begin{figure}[htb]
	\centering
	\begin{minipage}[c]{0.4\textwidth}
		\includegraphics[page=4, width=\textwidth]{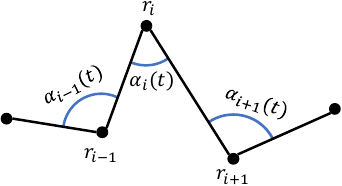}
	\end{minipage}
	\hfill
	\begin{minipage}[c]{0.55\textwidth}
		\captionof{figure}{A configuration in which \( r_{i} \) fulfills the first \emph{Angle Pattern}, i.e.; \( \alpha_{i}(t) \)is a local minimum.}
		\label{figure:angle-patterns}
	\end{minipage}
\end{figure}

\paragraph{Orientation Patterns}

A robot $r_i$ gets an \runInit{} if one of the following patterns is fulfilled.
\cref{figure:orientation-pattern-1,figure:orientation-pattern-2} depict the two patterns.

\begin{enumerate}[topsep=0pt,itemsep=-1ex,partopsep=1ex,parsep=1ex]
	\item The robot is between three angles that have a different orientation than $\alpha_{i}(t)$, i.e.\ $\sign_i(\alpha_{i-1}(t)) = \sign_i(\alpha_{i+1}(t)) = \sign_i(\alpha_{i+2}(t))  \neq  \sign_i(\alpha_{i}(t)) $ or $\sign_i(\alpha_{i-2}(t)) = \sign_i(\alpha_{i-1}(t)) = \sign_i(\alpha_{i+1}(t))  \neq  \sign_i(\alpha_{i}(t)) $

	\item The robot is at the border of a sequence of at least two angles with the same orientation next to a sequence of at least three angles with the same orientation, i.e.\ $\sign_i(\alpha_{i-1}(t)) = \sign_i(\alpha_{i}(t)) \neq \sign_i(\alpha_{i+1}(t)) = \sign_i(\alpha_{i+2}(t)) = \sign_i(\alpha_{i+3}(t))$ or $\sign_i(\alpha_{i+1}(t))$ $= \sign_i(\alpha_{i}(t))$ $\neq \sign_i(\alpha_{i-1}(t))$ $= \sign_i(\alpha_{i-2}(t)) = \sign_i(\alpha_{i-3}(t))$
\end{enumerate}

\begin{figure}[htb]
	\centering
	\begin{minipage}[t]{0.48\textwidth}
		\centering
		\includegraphics[page=5, width=0.8\textwidth]{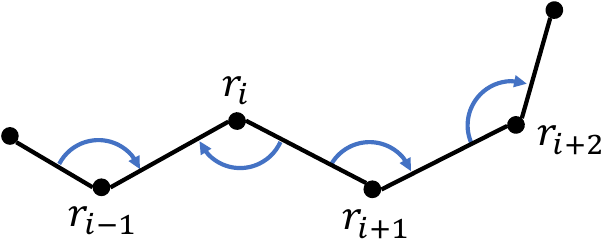}
		\captionof{figure}{A configuration in which \( r_{i} \) fulfills the first \emph{Orientation Pattern}.}
		\label{figure:orientation-pattern-1}
	\end{minipage}
	\hfill
	\begin{minipage}[t]{0.48\textwidth}
		\centering
		\includegraphics[page=6, width=0.85\textwidth]{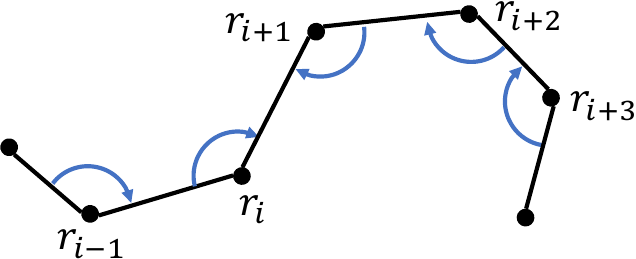}
		\captionof{figure}{A configuration in which \( r_{i} \) fulfills the third \emph{Orientation Pattern}.}
		\label{figure:orientation-pattern-2}
	\end{minipage}
\end{figure}
\paragraph{Vector Length Patterns}

If both Angle and Orientation Patterns fail, we consider vector lengths.
A robot \( r_{i} \) is assigned an \runInit{} if one of the following patterns is fulfilled, see \cref{figure:vector-pattern} for an example.
For better readability, we omit the time parameter $t$, i.e.\ we write $u_i$ instead of $u_i(t)$.
In the patterns, the term \emph{locally minimal} occurs.
$\|u_i\|$ is locally minimal means that all other vectors that can be seen by $r_i$ are either larger or have the same length.

\begin{enumerate}[topsep=0pt,itemsep=-3ex,partopsep=1ex,parsep=1ex]
	\item The robot is located at a locally minimal vector next to two succeeding larger vectors, i.e.\  \(\|u_i\|\) is locally minimal \textbf{and} \(\|u_{i-1}\| > \|u_i\| <\|u_{i+1}\|\) \textbf{and}
	\(\|u_i\| < \|u_{i+2}\|\) \textbf{or} \(\|u_{i+1}\|\) is locally minimal \textbf{and} \(\|u_{i}\| > \|u_{i+1}\| < \|u_{i+2}\|\) \textbf{and}
	\(\|u_{i+1}\| < \|u_{i+3}\|\). \\
	\item The robot is at the boundary of a sequence of at least two locally minimal vectors, i.e.\
	\(\|u_{i-1}\| = \|u_i\| < \|u_{i+1}\|\) \textbf{or} \(\|u_i\| > \|u_{i+1}\| = \     \|u_{i+2}\|\) \\
\end{enumerate}



\begin{figure}[htb]
	\centering
	\begin{minipage}[c]{0.4\textwidth}
		\includegraphics[page=7, width=\textwidth]{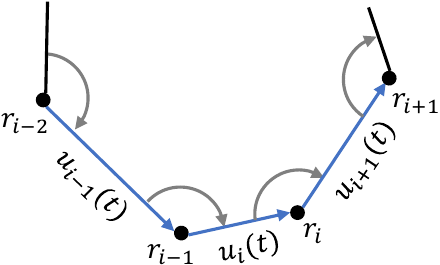}
	\end{minipage}
	\hfill
	\begin{minipage}[c]{0.55\textwidth}
		\caption{A configuration in which \( r_{i} \) (and also \(r_{i-1}\)) fulfills the first \emph{Vector Length Pattern}, i.e.; the length of \( u_{i}(t) \) is a local minimum.}
		\label{figure:vector-pattern}
	\end{minipage}
\end{figure}

\paragraph{How to start runs?}

Robots with \runInit{}s or \jointRunInit{}s try every $7$  rounds (counted with lights) to start new \runState{}s.
A robot $r_i$ with $\mathit{init}(r_i) = \mathit{true}$ only starts a new \runState{} in case it is not blocked and $\mathit{run}(N_i(t)) = \emptyset$ to ensure sufficient distance between \runState{}s.
Consider a robot $r_i$ with an \runInit{} and $\mathit{run}(N_i(t)) = \emptyset$.
$r_i$ only starts new \runState{}s provided $d(p_{i-1}(t), p_{i+1}(t)) > 1$.
Otherwise, it directly executes a \merge{}.
Given $d(p_{i-1}(t), p_{i+1}(t)) > 1$, $r_i$ generates two new \runState{}s at its direct neighbors with opposite directions as follows: $r_i$ executes a \shorten{} and generates two new \runState{}s
$\kappa_1$ and $\kappa_2$ with $r(\kappa_1,t+1) = r_{i+1}$, $r(\kappa_1,t+2) = r_{i+2}$ and, similarly, $r(\kappa_2,t+1) = r_{i-1}$, $r(\kappa_2,t+2) = r_{i-2}$.
Two robots $r_i$ and $r_{i+1}$ with a \jointRunInit{} proceed similarly:
given $d(p_{i-1}(t),$ $p_{i+2}(t)) \leq 2$, they directly execute a \jointMerge{}.
Otherwise $r_i$ and $r_{i+1}$ execute a \jointShorten{} and induce two new \runState{}s at their neighbors with opposite direction.
More formally, the \runState{}s $\kappa_1$ and $\kappa_2$ are generated with $r(\kappa_1,t+1) = r_{i+2}$, $r(\kappa_1,t+2) = r_{i+3}$, $r(\kappa_2,t+1) = r_{i-1}$ and $r(\kappa_2,t+2) = r_{i-2}$.

\subsection{Symmetric Algorithm}

As a consequence of the patterns in \Cref{section:asymmetric}, there is a set of configurations where no \runInit{} can be generated.
Intuitively, such configurations have no local criterion that identifies some robot as different from its neighbors, i.e.; they are \emph{symmetric}.
We start by defining precisely the class of configurations in which no \runState{} can be generated by our protocol.
Afterwards, we show how we can still gather such configurations.

\paragraph{Which configurations are left?}
There are some configurations in which every robot locally has the same view.
These configurations can be classified as the \emph{isogonal configurations}.
Intuitively, a configuration is isogonal if all angles have the same size and orientation and either all vectors have the same length or there are two alternating vector lengths.
In every other configuration, the patterns described in \Cref{section:asymmetric} detect the asymmetry.
The description of isogonal configurations is easier when considering polygons.
For some round \( t \), the set of all vectors $u_i(t)$ describes a polygon denoted as the \emph{configuration polygon} of round $t$.
A configuration is then called an \emph{isogonal configuration} in case its configuration polygon is an isogonal polygon.

\begin{definition}[\cite{grunbaum1994metamorphoses}]
	A polygon $P$ is \emph{isogonal} if and only if for each pair of vertices there is a symmetry of $P$ that maps the first onto the second.
\end{definition}

Grünbaum~\cite{grunbaum1994metamorphoses} classified the set of isogonal polygons.
Interestingly, this set of polygons consists of the set of \emph{regular star polygons} and polygons that can be obtained from them by a small translation of the vertices.

\begin{definition}[\cite{grunbaum1994metamorphoses}]\label{definition:schlaefli2}
	The \emph{regular star polygon} $\{n/d\}$ ($\{n/d\}$ is also denoted as the \emph{Schläfli symbol}) describes the following polygon with $n$ vertices:
	Consider a circle $C$ and fix an arbitrary radius $R$ of $C$.
	Place $n$ points $A_1, \dots, A_n$ such that $A_j$ is placed on $C$ and forms an angle of $2\pi \frac{d}{n} \cdot j$ with $R$ and connect $A_j$ to $A_{j+1} \textnormal{ mod } n$ by a segment.
\end{definition}

\begin{definition}
	A configuration is called a \emph{regular star configuration}, in case the configuration polygon is a regular star polygon.
\end{definition}

\begin{lemma}[\cite{grunbaum1994metamorphoses}]
	For $n$ odd, every isogonal polygon is a regular star polygon.
	For $n$ even, isogonal polygons that are not regular star polygons can be constructed as follows:
	Take any regular star polygon $\{n/d\}$ based on the circle \( C \) of radius \( R \).
	Now, choose a parameter $0 < t < \frac{n}{2}$ and locate the vertex $A_j$ such that its angle to $R$ is $\frac{2\pi}{n} \cdot \left(j \cdot d + (-1)^j \cdot t\right)$.  Choosing $t = \frac{n}{2}$ yields the polygon $\{n/d\}$ again.
	Larger values for $t$ obtain the same polygons as in the interval $[0, \frac{n}{2}]$.
\end{lemma}

Examples of such configurations can be seen in \cref{figure:isogonal-non-star} and \cref{figure:isogonal-regular-star}.

\begin{figure}[htb]
	\begin{minipage}[c]{0.47\textwidth}
		\centering
		\includegraphics[page=9, width=0.6\textwidth]{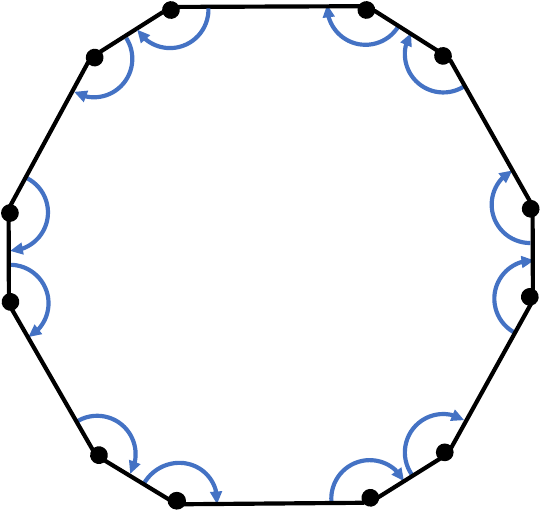}
		\captionof{figure}{An isogonal configuration that has two alternating vector lengths and all angles are the same with \( n = 12 \).}
		\label{figure:isogonal-non-star}
	\end{minipage}
	\hfill
	\begin{minipage}[c]{0.47\textwidth}
		\centering
		\includegraphics[page=10, width=0.6\textwidth]{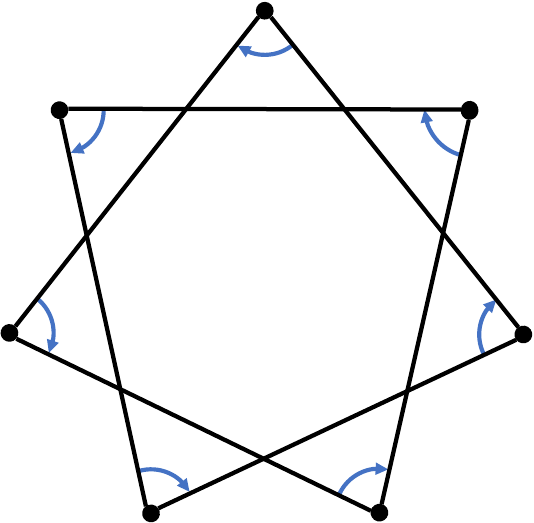}
		\captionof{figure}{An isogonal configuration of which the polygon is a star configuration with \( n = 7 \).}
		\label{figure:isogonal-regular-star}
	\end{minipage}
\end{figure}

\paragraph{Movement}

Consider a robot $r_i$.
It can observe its neighborhood and assume it is in an isogonal configuration if all $\alpha_{i}(t)$  in its neighborhood have the same size and orientation and either all vectors $u_i(t)$ have the same length or have two alternating lengths.
In case an isogonal configuration is assumed and $\mathit{init}(N_i(t)) = \emptyset{}$ and $\mathit{run}(N_i(t)) = \emptyset{}$ hold true, \( r_{i} \) performs on of the two following symmetrical operations.
In case all vectors $u_i(t)$ have the same length, it performs a \bisector{} as defined below.
The purpose of the \bisector{} is to move all robots towards the center of the surrounding circle.
Otherwise (in case of two alternating vector lengths), the robot executes a \schlaefli{}.
The goal of the \schlaefli{} is to transform an isogonal configuration with two alternating vector lengths into a regular star configuration such that \bisector{}s are applied afterwards.

\paragraph{\Bisector{}}
In the \emph{\bisector{}}, a robot $r_i$ computes the angle bisector of vectors pointing to its direct neighbors (bisecting the angle of size less than $\pi$) and jumps to the point $p$ on the bisector such that $d(p_{i-1}(t),p  ) = d(p_{i+1}(t),p) =~1$.
In case $d(p_i(t),p) > \bisectorConstant{}$, the robot moves only a distance of \bisectorConstant{} towards $p$.

\paragraph{\Schlaefli{}}
Let $C$ be the circle induced by $r_i$'s neighborhood and $R$ its radius.
In case the diameter of $C$ has length of at most $2$, $r_i$ jumps to the midpoint of $C$.
Otherwise, the robot $r_i$ observes the two circular arcs $L_{\alpha} = \alpha \cdot R$ and $L_{\beta} = \beta \cdot R$ connecting itself to its direct neighbors.
The angles $\alpha$ and $\beta$ are the corresponding central angles measured from the radius $R_i$ connecting $r_i$ to the midpoint of $C$.
W.l.o.g. assume $L_{\alpha} < L_{\beta}$.
$r_i$ jumps to the point on $L_{\beta}$ such that $L_{\alpha}$ is enlarged by $R \cdot (\frac{\beta-\alpha}{4})$ and $L_{\beta}$ is shortened by the same value.

\subsection{Combination}
The asymmetric and the symmetric algorithm are executed in parallel.
More precisely, robots whose neighborhood fulfills the property of being an isogonal configuration move according to the symmetric algorithm while all others follow the asymmetric algorithm.
To ensure that the two algorithms do not hinder each other, we need one additional rule here: the exceptional generation of \runInit{}s.
Intuitively, if some robots follow the asymmetric algorithm while others execute the symmetric algorithm, there are borders at which a robot $r_i$ moves according to the symmetric algorithm while its neighbor does not move at all (the neighbor cannot have a \runState{}, otherwise $r_i$ would not move according to the symmetric algorithm).
At these borders it can happen that the length of the chain increases.
To prevent this from happening too often we make use of an additional visible state.
Robots that move according to the symmetric algorithm store this in a visible state.
If they detect in the next round that this state is activated but their local neighborhood does not fulfill the criterion of being an isogonal configuration, they will conclude that their neighbor has not moved and thus the chain is not completely symmetric.
To ensure that this does not occur again, they assign an \runInit{} to themselves and thus add an additional source of asymmetry to the chain.

\section{Analysis}

Due to space limitations, we only introduce the high level idea of the analysis.
The proofs are deferred to \Cref{section:appendix}.
One of the crucial properties for the correctness of our algorithm is that it maintains the connectivity of the chain.
For the proof we show that the operations of isolated runs, joint run-pairs and the operations of the symmetrical algorithm do not break the connectivity as well as no other pattern of runs exists (for instance a sequence of three neighboring robots having a \runState{}).

\begin{restatable}{lemma}{connectivityLemma}
	A configuration that is connected in round $t$ stays connected in round $t+1$.
\end{restatable}

The asymmetric algorithm depends on the generation of \runState{}s.
We prove that in every asymmetric (non-isogonal) configuration in which no robot has an \runInit{}, at least one pattern is fulfilled and thus at least one \runInit{} exists in the next round.

\begin{restatable}{lemma}{runInitLemma}
	Given a configuration without any \runInit{} in round $t$.
	Either the configuration is isogonal or at least one \runInit{} exists in round $t+1$.
\end{restatable}

So far, we know that at least one \runInit{} exists in an asymmetric configuration.
Hence, \runState{}s are periodically generated.
The following lemma is the key lemma of the asymmetric algorithm: Every run that is started at robot $r_i$ will never visit $r_i$ again in the future:
The \runState{} either stops by a \merge{}, a \jointMerge{}, a \shorten{}, a \jointShorten{} or it is stopped by a \merge{} or a \jointMerge{} of a different \runState{}.
This implies that a run cannot execute $n$ succeeding \hop{}s or \jointHop{}s.
To see this, consider a \runState{} $\kappa$ in round $t$ with run vector $\vec{v}_{\kappa}$ located at robot $r_i$ with $r(\kappa,t+1) = r_{i+1}$.
W.l.o.g. assume that $\vec{v}_{\kappa}$ is parallel to the y-axis and points upwards in a global coordinate system and $p_{i}(t) = \left(0,0\right)^T$.
In case $r_i$ executes a \hop{} now (the arguments for a \jointHop{} are analogous), it must hold $y_{i+1}(t) > y_i(t)$ (because $\alpha_i(t) \geq \frac{7}{8} \pi$).
Thus, $\kappa$ moves continuously upwards in a global coordinate system.
Assume that $\kappa$ started at a robot $r_s$ in round $t_0$.
Given that $\kappa$ executes only \hop{}s and \jointHop{}s, there must be a round $t'$ with $t'-t_0 \leq n-1$ such that $r(\kappa,t'+1) = r_s$.
To execute one additional \hop{} or \jointHop{}, the robot $r_s$ must lie above of $r(\kappa,t)$.
This is impossible due to the threshold of $\frac{7}{8} \pi$ and the fact that $r_s$ can execute at most every second round an operation while $\kappa$ induces an operation in every round.
As a consequence, each \runState{} is stopped after at most $n-1$ rounds.

\begin{restatable}{lemma}{runStopLemma}\label{lemma:runStopLemma}
	A run does not visit the same robot twice.
	Before reaching the robot where it started, the run either stops by a \shorten{}, a \jointShorten{}, a \merge{} or a \jointMerge{} or it is stopped by a \merge{} or \jointMerge{} of a different run.
\end{restatable}

Next, we count the number of \runState{}s that are needed to gather all robots on a single point. \Cref{lemma:runStopLemma} states that each \runState{} stops either by a \shorten{}, \jointShorten{}, \merge{} or \jointMerge{} or it is stopped via a \merge{} or \jointMerge{} of a different \runState{}.
Obviously, there can be at most $n-1$ \merge{}s and \jointMerge{}s because these operations reduce the number of robots in the chain.
To count the number of \shorten{}s and \jointShorten{}s, we consider two cases: either the two vectors involved in a \shorten{} have both a length of at least $\frac{1}{2}$ or one vector is smaller and the other one is larger than $\frac{1}{2}$ (the case that both vectors are smaller would lead to a \merge{}).
Due to the threshold of $\frac{7}{8} \pi$, we can prove that the chain length reduces by at least a constant ($\shortenConstant$) in case both vectors have a length of at least $\frac{1}{2}$.
For the case that one vector is smaller and the other one larger than $\frac{1}{2}$, the chain length does not necessarily decrease by a constant.
Instead, the smaller vector increases and has a length of at least $\frac{1}{2}$ afterwards.
Hence, either the chain length or the number of small vectors decreases.
New small vectors can only be created upon the execution of \merge{}, \jointMerge{}s or \schlaefli{}s and \bisector{}s.
For each of the mentioned operations, we can prove that it is only executed a linear number of times.
All in all, we can conclude that a linear number of  \runState{}s is needed to gather all robots on a single point.

\begin{restatable}{lemma}{totalCountRunLemma} \label{lemma:totalCountRuns}
	At most $143\,n$ runs are required to gather all robots.
\end{restatable}

To conclude a final runtime for the asymmetric algorithm, we need to prove that many \runState{}s are generated.
This can be done by a witness argument.
Consider an \runInit{}.
Every $7$ rounds, this state either creates a new \runState{} or it sees other \runState{}s in its neighborhood and waits.
This way, we can count every $7$ rounds a new \runState{}: Either the robot with the \runInit{} starts a new \runState{} or it waits because of a different \runState{}.
Since a robot can observe $7$ neighbors, we count a new \runState{} every $7$ rounds.
Roughly said, we can prove that in $k$ rounds $ \approx \frac{k}{7}$ \runState{}s exist.
This holds until the \runInit{} is removed due to a \merge{} or \jointMerge{}.
Afterwards, we can continue counting at the next \runInit{} in the direction of the \runState{} causing the \merge{} or \jointMerge{}.
Combining all the arguments above leads to a linear runtime of the asymmetric algorithm.

\begin{restatable}{lemma}{numberOfRoundsLemma}
	A configuration that does not become isogonal gathers after at most $4018\,n$ rounds.
\end{restatable}

It remains to prove a linear runtime for the symmetric algorithm.
The symmetric algorithm consists of two parts: first an isogonal configuration is transformed into a regular star configuration and afterwards the robots move towards the center of the surrounding circle.
The transformation to a regular star configuration requires a single round in which all robots execute a \schlaefli{}.

\begin{restatable}{lemma}{isogonalToRegularStarLemma}
	In case the configuration is isogonal but not a regular star configuration at time $t$, the configuration is a regular star configuration at time $t+1$.
\end{restatable}

To prove a linear runtime for regular star configurations it is sufficient to analyze the runtime for the regular polygon $\{n/1\}$.
In all other regular star configurations, the inner angles are smaller and thus, the robots can move larger distances towards the center of the surrounding circle.
We use the radius of the surrounding circle as a progress measure.
Although the radius decreases very slow initially, it decreases by a constant in every round after a linear number rounds.
The linear runtime follows.

\begin{restatable}{lemma}{regularStarRoundLemma}
	Regular star configurations are gathered in  at most $30\,n$ rounds.
\end{restatable}

\bibliographystyle{splncs04}
\bibliography{closedChain}

\newpage
\appendix

\section{Omitted Proofs} \label{section:appendix}

\subsection{Connectivity}

\begin{lemma} \label{lemma:connectivityRunMovement}
	The movement operations of isolated runs and joint run-pairs keep the connectivity of the chain.
\end{lemma}

\begin{proof}
	For isolated \runState{}s, a robot $r_i$ executes only a \merge{} if $d\left(p_{i-1}(t), p_{i+1}(t)\right) \leq 1$.
	Since $r_i$ moves either to $p_{i-1}(t)$ or to $p_{i+1}(t)$ (depending on the direction of the \runState{}) and neither $r_{i-1}$ nor $r_{i+1}$ moves in the same round, the chain remains connected.

	Consider a robot $r_i$ that executes a \shorten{}. It moves to the midpoint between its neighbors, more formally $p_i(t+1) = \frac{1}{2} p_{i-1}(t) + \frac{1}{2} p_{i+1}(t)$.
	Since the configuration is connected in round $t$, it holds $d\left(p_{i-1}(t), p_{i+1}(t)\right) \leq 2$ and thus it follows  $d \left(p_{i-1}(t+1), p_i(t+1) \right) \leq 1$ and  $d \left(p_{i}(t+1), p_{i+1}(t+1) \right) \leq 1$ as $p_{i-1}(t+1) = p_{i-1}(t)$ and $p_{i+1}(t+1) = p_{i+1}(t)$.

	Now suppose $r_i$ executes a \hop{} in the direction of $r_{i+1}$.
	The \hop{} exchanges the vectors $u_i(t)$ and $u_{i+1}(t)$.
	Since both vectors have a length of at most $1$, the connectivity is ensured in round $t+1$.

	The arguments for a \jointMerge{}, \jointShorten{} and \jointHop{} are analogous.

\end{proof}

\begin{lemma} \label{lemma:connectivitySymmetricMovement}
	\Schlaefli{}s and \Bisector{}s keep the connectivity of the chain.
\end{lemma}

\begin{proof}
	For \bisector{}s, this follows directly from the definition:
	Two neighboring robots that both execute a bisector jump towards the center of the same circle and thus their distance decreases.
	Given that a robot executes a \bisector{} while its neighbor does not move at all, the definition of the \bisector{} ensures that their distance is at most $1$ in the next round.
	Other cases cannot occur: since \bisector{}s are only executed in case no \runState{} is visible in the neighborhood of a robot, it cannot happen that a neighbor executes a different operation.
	Thus, \bisector{}s maintain the connectivity of the chain.

	The last arguments also apply for the \schlaefli{}.
	A robot only executes a \schlaefli{} if either its neighbors also execute a \schlaefli{} or do not move at all.
	Suppose two neighboring robots execute a \schlaefli{}.
	In round $t$ they are connected via a circular arc $L_{\alpha} = \alpha \cdot r$ or $L_{\beta} = \beta \cdot r$, where $r$ denotes the radius of the circumcircle of the neighborhood.
	Assume that $\alpha < \beta$.
	If the robots are connected with $L_{\beta}$, the robots move closer to each other and maintain the connectivity.
	Otherwise, they move away from each other, but the new circular arc connecting the two robots is $L_{\alpha} + r \cdot \left(\frac{\beta-\alpha}{2}\right)$ which is less than $L_{\beta}$.
	Hence, the two neighboring have a distance of at most $1$ in the next round.
\end{proof}

\begin{lemma} \label{lemma:connectivityRunInitSequence}
	In every round $t$, there exist no sequence of neighboring robots of length at least $3$, all having an \runInit{}.
\end{lemma}

\begin{proof}
	Robots that already have neighbors with an \runInit{} do not generate new ones.
	Thus, we only have to look at sequences of robots in which no robot has an \runInit{} and show that for at most two neighbors a pattern is fulfilled.
	First, consider the angle patterns.
	Suppose for a robot $r_i$ the first angle pattern is fulfilled (the argumentation for the second angle pattern is analogous, because it describes the mirrored version.)
	Thus, it holds, $\alpha_{i-1}(t) > \alpha_{i}(t) \leq \alpha_{i+1}(t)$.
	For $r_{i-1}$, no angle pattern can hold because $\alpha_{i-1}(t)$ is no local minimum.
	Additionally, $r_{i-1}$ does not check any further patterns because $\alpha_{i-1}(t) \neq \alpha_{i}(t)$ and hence no full symmetry is given.
	The only case in which an angle pattern for $r_{i+1}$ can be fulfilled is the case that $\alpha_{i}(t) = \alpha_{i+1}(t)$.
	If $\alpha_{i+1}(t) \geq \alpha_{i+2}(t)$, no pattern for $r_{i+1}$ holds, since $\alpha_{i+1}(t)$ is either no local minimum or $\alpha_{i}(t) = \alpha_{i+1}(t) = \alpha_{i+2}(t)$.
	Thus, the only case that a pattern for both $r_i$ and $r_{i+1}$ holds is the case that $\alpha_{i-1}(t) > \alpha_{i}(t) = \alpha_{i+1}(t) < \alpha_{i+2}(t)$.
	Since $\alpha_{i+2}(t)$ is also no local minimum, no pattern for $r_{i+2}$ holds.
	As a consequence, it can happen that an angle pattern for both $r_{i}$ and $r_{i+1}$ is fulfilled, but then, neither $r_{i-1}$ nor $r_{i+2}$ fulfills a pattern.

	We continue with the orientation patterns.
	There are two classes of orientation patterns, we start with the first class.
	Assume that for $r_i$ an orientation pattern is fulfilled.
	Thus, it holds $\sign(\alpha_{i-1}(t)) = \sign(\alpha_{i+1}(t)) = \sign(\alpha_{i+2}(t)) \neq \sign(\alpha_{i}(t))$ (the other pattern in this class is a mirrored version and the same argumentation can be applied).
	No orientation pattern of the first class is fulfilled for $r_{i+1}$, because the neighboring angles have a different orientation.
	Additionally, no pattern of the second class can be fulfilled, because $\sign(\alpha_{i}(t)) \neq \sign(\alpha_{i-1}(t))$.
	For $r_{i-1}$ no pattern of the second class can be fulfilled because $\sign(\alpha_{i}(t)) \neq \sign(\alpha_{i+1}(t))$.
	A pattern of the first class can only be fulfilled if $\sign(\alpha_{i-3}(t)) = \sign(\alpha_{i-2}(t)) = \sign(\alpha_{i}(t))$.
	In this case, no pattern for $r_{i-2}$ can be fulfilled because its neighboring angles have a different orientation and it is not located at the boundary of two sequences of angle orientations of length at least two.
	Hence, it can happen that a pattern for $r_i$ and $r_{i-1}$ is fulfilled but then neither a pattern for $r_{i-2}$ nor for $r_{i+1}$ is fulfilled.

	Lastly, we consider the vector patterns.
	Suppose that for $r_i$ a vector pattern is fulfilled.
	We give the arguments for the first class, the second class is analogous.
	Arguments for the third class are given afterwards.
	Assume that the first vector pattern is fulfilled for $r_i$ and thus
	\(\|u_i(t)\|\) is locally minimal \textbf{and} \(\|u_{i-1}(t)\|> \|u_i(t)\| < \|u_{i+1}(t)\|\) \textbf{and} \(\|u_i(t)\|< \|u_{i+2}(t)\|\).
	In this case, no pattern can be fulfilled for $r_{i+1}$ because \(\|u_{i+1}(t)\|\) and \(\|u_{i+2}(t)\|\) are not locally minimal (\(\|u_i(t)\|\) is smaller).
	For $r_{i-1}$ at most the second pattern can be fulfilled (because \(|u_{i-1}(t)\|\) is not locally minimal.
	The second pattern can only be fulfilled if $\|u_{i-2}(t)\|> |u_i(t)\|$.
	In this case, no pattern for $r_{i-2}$ can hold because neither $\|u_{i-2}(t)\|$ nor $\|u_{i-1}(t)\|$ are locally minimal.
	Thus, a pattern for $r_{i-1}$ and $r_i$ can be fulfilled, but no patterns for $r_{i-2}$ and $r_{i+1}$.
	The arguments for the second class are analogous.

	Now suppose that for a robot $r_i$ the third pattern is fulfilled.
	Thus, it holds: \(\|u_{i-1}(t)\|= \|u_{i}(t)\| < \|u_{i+1}(t)\| \)  (the arguments for the other pattern are analogous because it describes the mirrored version).
	In this case, no vector pattern holds for $r_{i-1}$ since both neighboring vectors have the same length.
	For $r_{i+1}$, the third pattern can be fulfilled.
	It must hold $\|u_{i+2}(t)\|= \|u_{i+3}(t)\| < \|u_{i+1}(t)\|$.
	However, the third pattern cannot hold for $r_{i+2}$ because both neighboring vectors have the same length.
	In addition, neither the first nor the second pattern can hold for $r_{i-1}$ and $r_{i+1}$ (due to the definition of the third pattern) and thus both $r_i$ and $r_{i+1}$ can generate an \runInit{} but neither $r_{i-1}$ nor $r_{i+2}$.

\end{proof}

\begin{definition}[Prohibited Run-Sequence]
	Three \runState{}s are called a \emph{prohibited run-sequence} if the three runs are located at three directly neighboring robots.
\end{definition}

\begin{definition}[Conflicting Run-Pair]
	Consider two \runState{}s $\kappa_1$ and $\kappa_2$.
	Two \runState{}s $\kappa_1, \kappa_2$ are called an \emph{opposite conflicting run-pair} in case $r(\kappa_1,t)$ and $r(\kappa_2,t)$ are direct neighbors and $r(\kappa_1,t+1) \neq r(\kappa_2,t)$ and $r(\kappa_2,t+1) \neq r(\kappa_1,t)$.
	Two \runState{}s $\kappa_1, \kappa_2$ are called an \emph{uni-directional conflicting run-pair} in case $r(\kappa_1,t+1) = r(\kappa_2,t)$ and $r(\kappa_2,t+1) \neq r(\kappa_1,t)$ or vice versa.
\end{definition}

\begin{definition}
	A configuration is called to be \emph{\runState{}-valid} in round $t$ if neither a prohibited run-sequence nor a conflicting run-pair exists.
\end{definition}

\begin{lemma} \label{lemma:connectivityJointRunPairs}
	Consider a run-valid configuration in round $t$ with a joint run-pair of \runState{}s $\kappa_1$ and $\kappa_2$ with $r(\kappa_1,t) = r_i$ and $r(\kappa_2,t) = r_{i+1}$.
	Suppose the robots execute a \jointHop{}.
	Then, $\mathit{run}(r_i,t+1) = \mathit{run}(r_{i+1},t+1) = \mathit{false}$.
\end{lemma}

\begin{proof}
	Since the configuration is run-valid it holds $\mathit{run}(r_{i-1},t) = \mathit{run}(r_{i+2},t) = \mathit{false}$.
	As the robots execute a \jointHop{} it holds $r(\kappa_1,t+1) = r_{i+2}$ and $r(\kappa_2,t+1) = r_{i-1}$.
	Thus, it can only happen that \runState{}s different from $\kappa_1$ and $\kappa_2$ are located at $r_i$ or $r_{i+1}$ in round $t+1$.
	Next, we argue that this is impossible.
	We prove that no run $\kappa_3$ that is located at a robot with index larger than $i+2$ can be located at $r_i$ or $r_{i+1}$, the arguments for \runState{}s located at robots with smaller indices are analogous.
	Assume that there exists an isolated \runState{} $\kappa_3$ with $r(\kappa_3,t) = r_{i+3}$.
	Depending on its direction it either holds $r(\kappa_3,t+1) = r_{i+4}$ or $r(\kappa_3,t+1) = r_{i+2}$.
	It follows that $\kappa_3$ cannot be located at $r_{i+1}$ or $r_i$.
	Similar arguments holds for isolated \runState{}s located at robots with larger indices.
	Now assume that there is a joint run-pair $\kappa_3$ and $\kappa_4$ with $r(\kappa_3,t) = r_{i+3}$ and $r(\kappa_4,t) = r_{i+4}$.
	It holds $r(\kappa_3,t+1) = r_{i+5}$ and $r(\kappa_4,t+1) = r_{i+2}$.
	Again, the same arguments hold for joint run-pairs located at robots with higher indices.
	It follows that $\mathit{run}(r_i,t+1) = \mathit{run}(r_{i+1},t+1) = \mathit{false}$.
\end{proof}

\begin{lemma} \label{lemma:connectivityRunValidity}
	A configuration that is \runState{}-valid in round $t$ is also \runState{}-valid in round $t+1$.
\end{lemma}

\begin{proof}

	\Cref{lemma:connectivityRunInitSequence} ensures that starting of new \runState{}s always ensures that no prohibited run-sequence exists.
	Beyond that, a \merge{} or a \jointMerge{} stops all \runState{}s in its neighborhood such that \runState{} do not come too close such that these operations also ensure that no prohibited run-sequence exists.
	\Shorten{}s and \jointShorten{}s stop the involved \runState{}s and do not change the number of robots in the chain.
	Hence, no prohibited run-sequences can be generated.
	\Hop{}s are only executed by isolated \runState{}s and continue in their run direction and thus can also not create prohibited run sequences.
	The only operation, we have to consider in more detail is the \jointHop{} because the involved \runState{}s skip the next robot in run direction and move to the next but one robot.
	Let $\kappa_1$ and $\kappa_2$ denote a joint run-pair with $r(\kappa_1,t) = r_i$ and $r(\kappa_2,t) = r_{i+1}$.
	By the definition of a joint run-pair it holds $\mathit{run}(r_{i-1},t) = \mathit{run}(r_{i+2},t) = \mathit{false}$.
	Moreover, \Cref{lemma:connectivityJointRunPairs} states that $\mathit{run}(r_i,t+1) = \mathit{run}(r_{i+1},t+1) =\mathit{false}$.
	In the following, we prove that $\kappa_1$ is not part of a prohibited run-sequence in round $t+1$, the arguments for $\kappa_2$ are analogous.
	By definition, it holds $r(\kappa_1,t+1) = r_{i+2}$.
	\Cref{lemma:connectivityJointRunPairs} gives us also that a prohibited run sequence cannot be generated by a joint run pair located at $r_{i+3}$ and $r_{i+4}$ or $r_{i+4}$ and $r_{i+5}$ since in both cases no run is located at $r_{i+4}$ in round $t+1$.
	Joint run pairs with larger indices are too far away to generate a prohibited run sequence.

	Additionally, this cannot happen by isolated \runState{}s because no neighboring robots have \runState{}s and they move to the next robots.
	Hence, no prohibited run sequence can exists in round $t+1$.

	Next, we argue that no conflicting run-pair exists in round $t+1$.
	The start of new runs never creates new conflicting run-pairs.
	Thus, in case an opposite or a conflicting run-pair exists in round $t+1$, both involved runs have already existed in round $t$.
	Consider an unidirectional run-pair at round $t+1$.
	In round $t$ both runs must have had a distance of at least $2$ (one robot without a run in between).
	The distance between the two robots can only decrease based on a \merge{} or a \jointHop{}.
	A \merge{} stops all runs in the neighborhood and cannot create such a run-pair.
	A \jointHop{} can also not create a uni-directional run-pair since the configuration has been run-valid in round $t$ and both neighboring robots have not had a \runState{}.
	Thus, no uni-directional run-pair can exist in round $t+1$.

	Assume now that an opposite run-pair exists at round $t+1$.
	This can only be the case if the two runs have been heading towards each other in round $t$.
	However, \jointHop{}s, \jointShorten{}s and \jointMerge{}s ensure that no opposite run-pair exist in round $t+1$.
	Hence, the configuration remains run-valid.

\end{proof}

\connectivityLemma*

\begin{proof}
	By \Cref{lemma:connectivityRunValidity} it holds that only isolated \runState{}s or joint run-pairs exists (or no \runState{} at all).
	\Cref{lemma:connectivityRunMovement} states that isolated \runState{} and joint run-pairs keep the connectivity of the chain and by \Cref{lemma:connectivitySymmetricMovement} this also holds for all operations of the symmetric algorithm.
	The lemma follows.
\end{proof}
\subsection{Asymmetric Case}

For the asymmetric case, we start with proving that in every asymmetric configuration at least one \runInit{} exists.

\runInitLemma*

\begin{proof}
	Assume that the configuration is not isogonal.
	Now suppose that not all angles $\alpha_{i}(t)$ are identical and
	consider the globally minimal angle $\alpha_{min}(t)$ at the robot $r_{min}$ (or any of them if the angle is not unique).
	The robot $r_{min}$ generates an \runInit{} if at least one of the neighboring angles is larger.
	Since $\alpha_{min}(t)$ is minimal, the only situation in which $r_{min}$ does not generate an \runInit{} is that $\alpha_{min-1}(t) = \alpha_{min}(t) = \alpha_{min+1}(t)$.
	In this case, follow the chain in any direction until a robot $r_{min}'$ is reached such that the next robot has a larger angle.
	Such a robot exists since we have assumed that not all angles are identical.
	For this robot, an angle pattern is fulfilled.
	As a consequence, given a configuration in which not all angles are identical, at least one \runInit{} is generated.

	Next, we consider the case that all angles are identical but not all angles have the same orientation.
	Observe first that the chain must contain two more angles of one orientation than of the other because the chain is closed.
	This essentially implies that the orientations cannot be alternating along the entire chain and it also cannot happen that alternating sequences of two angle orientations exists.
	More formally, the chain cannot consist only of the following two sequences:

	\begin{enumerate}
		\item $\sign_i(\alpha_{i}) \neq \sign_i(\alpha_{i+1}) \neq \sign_i(\alpha_{i+2}), \dots$
		\item  $\sign_i(\alpha_{i}) = \sign_i(\alpha_{i+1}) \neq \sign_i(\alpha_{i+2})  = \sign_i(\alpha_{i+3}) \neq \sign_i(\alpha_{i+4}) , \dots$
	\end{enumerate}

	As a consequence, at least one of the orientation patterns must be fulfilled: either there exists a sequence of at least three angles with the same orientation and a pattern is fulfilled at the boundary of such a sequence or there exists a robot $r_i$ that lies between three angles with a different orientation than $\alpha_{i}(t)$.
	Hence, given that all angles have the same size but not the same orientation, there must be at least one fulfilled orientation pattern.

	Lastly, we take a look at the vector length patterns.
	Now we assume that all angles in the chain have the same size and the same orientation.
	Since we assume that the configuration is not isogonal, not all vectors can have the same length and it also cannot be the case that there exist only two different vector lengths that are alternating along the chain.
	Consider the vector of global minimal length $u_{min}$ (or any of them if the length is not unique).
	Two cases can occur: either a sequence of at least two neighboring vectors of length $\|u_{min}\|$ exists: at the end of such a sequence, the third vector pattern is fulfilled.
	In case no such sequence exists, all vectors having the length of $\|u_{min}\|$ have direct neighbors that are larger.
	Since the configuration is not isogonal, there must be a vector of length $\|u_{min}\|$ such that the direct neighboring vectors are larger and at least one of the next but one vectors is also larger (otherwise the configuration is isogonal with two alternating vector lengths).
	At such a robot the first or second pattern is fulfilled.

	All in all, we have proven that for configuration that are not isogonal at least one pattern is fulfilled.
\end{proof}

Next, we count the number of occurrences of \shorten{}s, \jointShorten{}s, \merge{}s and \jointMerge{}s until all robots are gathered.
Since each \merge{} and \jointMerge{} reduces the number of robots in the chain, the following lemma trivially holds.

\begin{lemma} \label{lemma:mergeCount}
	There are at most $n-1$ \merge{}s and \jointMerge{}s.
\end{lemma}

There are basically two types of occurrences of \shorten{}s.
Either both involved vectors have a length of at least $\frac{1}{2}$ or one vector is larger than $\frac{1}{2}$ while the other one is shorter.
In the first case, we prove that the chain length decreases by a constant.
The second case reduces the number of vectors of length less than $\frac{1}{2}$ in the chain.

\begin{lemma} \label{lemma:shortenLargeProgress}
	Assume an isolated run $\kappa$ with $r(\kappa,t) = r_i$ and $r(\kappa,t+1)$ executes a \shorten{} and both $\|u_i(t)\| \geq \frac{1}{2}$ and $\|u_{i+1}(t)\|\geq \frac{1}{2}$.
	Then $L(t+1) = L(t) - \shortenConstant$.
\end{lemma}

\begin{proof}
	$u_i(t)$ and $u_{i+1}(t)$ denote the two vectors involved in the \shorten{}, $a = \|u_i(t) \|$ and $b = \|u_{i+1}(t)\|$.
	Additionally, $c = \|u_i(t) + u_{i+1}(t)\|$.
	The length of the chain decreases by $a+b-c$.
	By the law of cosines, $c = \sqrt{a^2 + b^2 - 2ab \cdot \cos \left(\alpha_{i}(t)\right)}$.
	The value of $c$ is maximized for $\alpha_{i}(t) = \frac{7}{8}\pi$.
	Hence, $a+b-c \geq a+b - \sqrt{a^2+b^2 -2ab \cdot \cos \left(\frac{7}{8}\pi\right)} = a+b - \sqrt{a^2+b^2 + 2ab \cdot \cos \left(\frac{\pi}{8}\right)}$.
	With boundary conditions $\frac{1}{2} \leq a \leq 1$ and $\frac{1}{2} \leq b \leq 1$, $a+b - \sqrt{a^2+b^2 + 2ab \cdot \cos \left(\frac{\pi}{8}\right)}$ is minimized for $a=b=\frac{1}{2}$.
	Thus, $a+b - \sqrt{a^2+b^2 + 2ab \cdot \cos \left(\frac{\pi}{8}\right)} \geq 1- \sqrt{\frac{1}{2} \cdot \left(1+ \frac{\sqrt{2+\sqrt{2}}}{2}\right)} \geq 0.019$.
\end{proof}

\JointShorten{}s are different, in a sense that every \jointShorten{} reduces the length of the chain by at least a constant.
This is because every run vector has a length of at least $\frac{1}{2}$ and thus two involved vectors in the \jointShorten{} have a length of at least $\frac{1}{2}$.

\begin{lemma} \label{lemma:jointShortenProgress}
	Assume that a joint run-pair executes a \jointShorten{} in round $t$.
	Then, $L(t+1) \leq L(t) - \shortenConstant$.
\end{lemma}

\begin{proof}
	Let $\kappa_1$ and $\kappa_2$ be the two \runState{}s with $r(\kappa_1,t) = r_i$ and $r(\kappa_2,t) = r_{i+1}$.
	The involved vectors are $u_i(t)$, $u_{i+1}(t)$ and $u_{i+2}(t)$.
	For simplicity, $a = \|u_i(t)\|, b=\|u_{i+1}(t)\|$, $c =\|u_{i+2}(t)\|$ and $d =\|u_i(t)+u_{i+1}(t) + u_{i+2}(t)\|$.
	The length of the chain decreases by $a+b+c-d$.
	By the triangle inequality, it follows $d \leq \|u_i(t) + u_{i+2}(t)\| + \|u_{i+1}(t)\|$.
	Thus, $a+b+c-d \geq a+b - \|u_i(t) + u_{i+2}(t)\|$.
	Now we can apply the same calculations as in the proof of \Cref{lemma:shortenLargeProgress} (since both $\|u_i(t)\| \geq \frac{1}{2}$ and $\|u_{i+2}(t)\| \geq \frac{1}{2}$ because all run vectors have a length of the least $\frac{1}{2}$) and obtain $L(t+1) \leq L(t) - \shortenConstant$.
\end{proof}

Next, we count the total number of \shorten{}s in which both involved vectors have a length of at least $\frac{1}{2}$ and \jointShorten{}s.
We have to take care here that the length of the chain might be increased by different operations.
This, however can only happen in a single case: a robot executes a \bisector{} while its neighbor does not move at all.
The exceptional generation of  \runInit{}s ensures that this happens at most $n$ times.

\begin{lemma} \label{lemma:shortenSmall}
	There are at most $\frac{n \cdot \left(1+\frac{1}{5}\right)}{\shortenConstant}$ executions of \shorten{}s in which both involved vectors have a length of at least $\frac{1}{2}$ and \jointShorten{}s.
\end{lemma}

\begin{proof}
	There is only one case in which $L(t)$ can increase: if a robot executes a \bisector{} while its direct neighbor does not.
	Since the maximal distance moved in a \bisector{} is $\frac{1}{5}$, $L(t)$ can increase by at most $\frac{1}{5}$ in this case.
	This, however, can happen at most $n$ times.
	To see this, observe that if a robot $r_i$ executes a \bisector{} in round $t$ and one of its neighbors does not, $\mathit{init}(r_i) = \mathit{true}$ in round $t+1$.
	The robot $r_i$ will not execute any further \bisector{} until it executes a \merge{} since an \runInit{} in the neighborhood of a robot prevents the robot from executing a \bisector{}
	and the \runInit{} is only removed after a \merge{},
	Thus, such a case can happen for at most $n$ times and thus, $L(t)$ is upper bounded by $n \cdot \left(1+ \frac{1}{5}\right)$.
	\Cref{lemma:shortenLargeProgress,lemma:jointShortenProgress} state that each \shorten{} in which both involved vectors have a length of at least $\frac{1}{2}$ and each \jointShorten{} decrease the length of the chain by at least $\shortenConstant$.
	Consequently, the total number of such operations can be upper bounded by $\frac{n \cdot \left(1+\frac{1}{5}\right)}{\shortenConstant}$.
\end{proof}

To count the total number of \shorten{}s required to gather all robots, we count the number of \shorten{}s in which one vector has a length of at most $\frac{1}{2}$ as a last step.

\begin{lemma}
	There are at most $4n$ executions of \shorten{}s such that one of the participating vectors has length less than $\frac{1}{2}$.
\end{lemma}

\begin{proof}
	There are at most $n$ vectors of length at most $\frac{1}{2}$ in the beginning.
	Every \shorten{} in which one vector of size at most $\frac{1}{2}$ and the other vector of length at least $\frac{1}{2}$ is involved, increases the length of the smaller vector to at least $\frac{1}{2}$.
	The only way to create new vectors of length less than $\frac{1}{2}$ is via a \merge{}, a \jointMerge{}, a \bisector{} or a \schlaefli{}.
	\Merge{}s and \jointMerge{}s are executed at most $n$ times and thus at most $n$ vectors of length less than $\frac{1}{2}$ can be generated.

	In the following, we consider the \bisector{} and the \schlaefli{}.
	It can happen that a robot executes a \bisector{} or a \schlaefli{} while its neighbor does not.
	In this case, the robot generates a new \runInit{}.
	Thus, this happens at most $n$ times since \runInit{}s prevent the robots in the neighborhood from executing \bisector{}s or \schlaefli{}s and an \runInit{} is only removed after a \merge{}.
	Therefore, at most $n$ vectors of length less than $\frac{1}{2}$ can be created by \bisector{}s or \schlaefli{}s if a neighboring robot does not execute the same operation.

	It remains to consider the case in which a robot and both its direct neighbors execute a \bisector{} and a \schlaefli{}.
	Consider now the \bisector{}.
	This operation only takes place at a robot $r_i$ in case $\|u_{i-1}(t)\|= \|u_i(t)\| = \|u_{i+1}(t)\| = \|u_{i+2}(t)\|$.
	Assume now that $\|u_i(t)\| > \frac{1}{2}$.
	Thus, in case $\|u_i(t+1)\| \leq \frac{1}{2}$ it also holds $\|u_{i+1}(t+1)\| \leq \frac{1}{2}$.
	Now either the configuration is completely isogonal, then no \shorten{} will be executed at $r_i$ anymore or at some parts of the chain still \runState{}s are generated.
	In the latter case, $r_i$ (and maybe also $r_{i-1}$ and $r_{i+1}$) can execute a \merge{} (since $d(p_{i-1}(t+1),p_{i+1}(t+1)) \leq 1$ as $\|u_i(t+1)\| \leq \frac{1}{2}$ and $\|u_{i+1}(t+1)\| \leq \frac{1}{2}$) such that the next \runState{} that comes close either executes a \merge{} at $r_{i-1}, r_i$ or $r_{i+1}$.
	Thus, also this case can happen at most $n$ times such that at most $2n$ vectors of length at most $\frac{1}{2}$ can be generated of which at most $n$ can be part of a future \shorten{}.

	Similar arguments apply for the \schlaefli{}:
	If both $\|u_i(t+1)\| < \frac{1}{2}$ and $\|u_{i+1}(t+1)\| < \frac{1}{2}$ after the \schlaefli{} at least one of them had a length of less than $\frac{1}{2}$ in round $t$.
	Hence, at most one vector of length less than $\frac{1}{2}$ can be generated by a \schlaefli{}.
	However, the same arguments as for the \bisector{} hold now:
	$r_i$ (and maybe also $r_{i-1}$ and $r_{i+1}$) can execute a \merge{} (since $d(p_{i-1}(t+1),p_{i+1}(t+1)) \leq 1$ as $\|u_i(t+1)\| \leq \frac{1}{2}$ and $\|u_{i+1}(t+1)\| \leq \frac{1}{2}$) such that the next \runState{} that comes close either executes a \merge{} at $r_{i-1}, r_i$ or $r_{i+1}$.
	Thus, before $r_i$ can generate a further vector of length at most $\frac{1}{2}$ via a further \schlaefli{}, either $r_{i-1}, r_i$ or $r_{i+1}$ execute a \merge{} such that this can happen also at most $n$ times.

	In total, we obtain at most $4n$ \shorten{}s in which one participating vector has a length of at most $\frac{1}{2}$: $n$ initial vectors that can have a length of at most $\frac{1}{2}$, $n$ vectors that can be generated via \merge{}s, $n$ vectors that can be generated via \schlaefli{}s and $n$ vectors that can be generated via \bisector{}s.
\end{proof}

Next, we prove that every \runState{} is stopped after at most $n$ rounds.
More precisely, no \runState{} visits the same robot twice.
For the proof, we state some auxiliary lemmata that analyze how the position of a robot changes in a global coordinate system based on the movement operations of the algorithm.

\begin{lemma} \label{lemma:mergeLemma}
	Assume that a robot $r_i$ executes a \merge{} or a \jointMerge{} in round $t$.
	Then, $y_i(t+1) \geq y_i(t) -1$.
\end{lemma}

\begin{proof}
	Consider a run $\kappa$ with $r(\kappa,t) = r_i$ and $r(\kappa,t+1) = r_{i+1}$.
	In the worst case, it holds $\alpha_{i}(t) = 0$ and $y_{i+1}(t) = y_i(t) -1$ such that $r_i$ moves to the position of $r_{i+1}$ and executes a \merge{} there.
	In a \jointMerge{} the distance is even less since the robots \merge{} in the midpoint between their neighbors.
\end{proof}

\begin{lemma} \label{lemma:shortenLemma}
	Assume that a robot $r_i$ executes a \shorten{} or a \jointShorten{} in round $t$.
	Then, $y_i(t+1) \geq y_i(t) - \frac{\sqrt{3}}{2}$.
\end{lemma}

\begin{proof}
	Observe first that $\alpha_i(t) > \frac{\pi}{3}$, otherwise a \merge{} can be executed.
	Thus, to maximize the distance covered in vertical direction consider the three involved robots to be the vertices of an equilateral triangle with side length $1$.
	The height of this triangle is $\frac{\sqrt{3}}{2}$.
	Thus, $y_i(t+1) \geq y_i(t) - \frac{\sqrt{3}}{2}$.
	The same holds for a \jointShorten{}.
\end{proof}

\begin{lemma} \label{lemma:hopLemma}
	Assume that a robot $r_i$ executes a \hop{} or a \jointHop{} in round $t$.
	Then, $y_i(t+1) \geq y_i(t) -\frac{1}{2}$.
\end{lemma}

\begin{proof}
	Observe first that every run-vector has a length of at least $\frac{1}{2}$, otherwise the robot that initiated the \runState{} would have immediately executed a \merge{}.
	Assume now that $r_i$ has a run $\kappa$ with $r(\kappa,t+1) = r_{i+1}$ and $r_i$ executes a \hop{}.
	The largest distance to cover in vertical direction for $r_i$ is $\frac{1}{2}$, in case $\|v_{\kappa}\| = \frac{1}{2}$and $y_{i+1}(t)  = y_i(t) -1$.
	Larger vectors $v_{\kappa}$ lead to a smaller distance moved in vertical direction, in case $\|v_{\kappa}\| > \|u_{i+1}(t)\|$, $r_i$ moves even upwards.
	Smaller vectors $u_i(t+1)$ have the same effect.
	Thus, $y_i(t+1) \geq y_i(t) - \frac{1}{2}$.
	The same holds for a \jointHop{}.
\end{proof}

\runStopLemma*

\begin{proof}
	Assume that the robot $r_i$ starts two runs $\kappa_1$ and $\kappa_2$ in round $t_0$.
	In round $t_0+1$, $r_i$ is located on the midpoint between $r_{i-1}$ and $r_{i+1}$.
	W.l.o.g.\ assume $r_i$ to be located in the origin of a global coordinate system with $r_{i+1}$ to be located on the $y$-axis above $r_i$ and $r_{i-1}$ to be located on the $y$-axis below $r_i$.
	The two run vectors are denoted as $v_{\kappa_1} = u_{i+1}(t_0+1)$ and $v_{\kappa_2} = - u_i(t_0+1)$ with $v_{\kappa_1} = - v_{\kappa_2}$.
	We prove exemplary for $\kappa_2$ that it stops before reaching $r_i$ again.
	The arguments for $\kappa_1$ are analogous.

	Suppose $\kappa_2$ does not stop due to a \shorten{}, \jointShorten{}, a \merge{} or a \jointMerge{}.
	This implies that $\kappa_2$ can only execute \hop{}s or \jointHop{}s in every round.
	Consider a round $t'$ with $r(\kappa_2,t') = r_j$ and $r(\kappa_2,t'+1) = r_{j-1}$.
	If $r_j$  executes a \hop{} or \jointHop{}, it holds $\alpha_{j}(t) > \frac{7}{8} \pi$.
	Thus, $y_{j-1}(t') < y_{j}(t')$.
	Since this holds for every \hop{} and \jointHop{} executed by $\kappa_2$ it also holds $y_{j-1}(t') < y_i(t) = 0$.
	We now bound $y_{\kappa_2}(t'+2)$:
	Since only hops are performed it must hold that the length of the next two vectors together is at least $1$ (otherwise the run stops and a \merge{} is executed).
	It follows that $y_{\kappa_2}(t'+2) \leq y_{\kappa_2}(t') - \cos \left(\frac{\pi}{8}\right) = y_{\kappa_2}(t') - \frac{\sqrt{2+\sqrt{2}}}{2}$.
	Thus, in every second round $y_{\kappa_2}(t)$ decreases by at least $\frac{\sqrt{2+ \sqrt{2}}}{2}$.

	We now compare the movements of $r_i$:
	At the same time $r_i$ can at most execute every second round an operation based on a run (since runs have a distance of $2$).
	In case of \shorten{}s or \jointShorten{}s, $y_i(t) \geq y_i(t) - \frac{\sqrt{3}}{2} > y_i(t) - \frac{\sqrt{2 + \sqrt{2}}}{2}$ (\Cref{lemma:shortenLemma}).
	In case of \hop{}s or \jointHop{}s, $y_i(t) \geq y_i(t) - \frac{1}{2}  > y_i(t) - \frac{\sqrt{2 + \sqrt{2}}}{2}$ (\Cref{lemma:hopLemma}).
	Thus, in case $r_i$ executes \shorten{}s, \jointShorten{}s, \hop{}s or \jointHop{}s, $\|y_i(t) - y_{\kappa_2}(t)\|$ decreases every second round by a constant.
	The same holds for \bisector{}s and \schlaefli{}s:
	By its definition, $y_i(t+1) \geq y_i(t) - \frac{1}{5}$ in case of a \bisector{}.
	For a \schlaefli{} it holds $y_i(t+1) \geq y_i(t) - \frac{1}{2}$ (since a robot jumps to the midpoint of two circular arcs).
	However, no two consecutive \schlaefli{}s can be executed since either the configurations is a regular star afterwards or the robot detects the asymmetry and does not execute a further \schlaefli{} (and generates an \runInit{} instead).
	As a consequence, also in case of \bisector{}s and \schlaefli{}s, $\|y_i(t) - y_{\kappa_2}(t)\|$ decreases by a constant.

	It remains to argue about \merge{}s and \jointMerge{}s.
	It can happen that $y_i(t)$ decreases by $1$ (\Cref{lemma:mergeLemma}).
	However, since all runs in distance $4$ are stopped and no further run in the neighborhood of $r_i$ is started within the next $4$ rounds, this can at most happen every fourth round.

	In the same time, however, $y_{\kappa_2}(t)$ decreases by at least $2 \cdot \frac{\sqrt{2+\sqrt{2}}}{2} = \sqrt{2+\sqrt{2}} > 1.8$.

	Thus, the only time in which $y_{\kappa_2}(t) > y_i(t)$ can hold is within the first $4$ rounds after the runs $\kappa_1$ and $\kappa_2$ have been started, more precisely only in the rounds $t_0+3$ and $t_0+4$ since $r_i$ can execute its first \merge{} or \jointMerge{} earliest in round $t_0+2$.
	For every round $t' > t_0+4$ it always holds $y(\kappa_2,t') < y_i(t')$.
	Hence, $\kappa_2$ cannot reach $r_i$ in a round $t' > t_0 + 4$ since $\kappa_2$ can only execute \hop{}s or \jointHop{}s and to execute a further hop when located at $r_{i+1}$ it must hold $y_i(t') < y_{\kappa_2}(t')$ which is a contradiction.
	In case $\kappa_2$ reaches $r_i$ again in round $t_0+3$ or $t_0+4$, the chain only has $5$ robots left such that the robots move towards the center of the smallest enclosing circle of all robots and all runs are stopped.

\end{proof}

\totalCountRunLemma*

\begin{proof}
	At most $n$ merges or \jointMerge{}s can happen (\Cref{lemma:mergeCount}).
	Additionally, there can be at most $\frac{n \cdot \left(1+ \frac{1}{5}\right)}{0.019}$ runs that stop with a \shorten{} or \jointShorten{} and two vectors of length at least $\frac{1}{2}$ (\Cref{lemma:shortenLargeProgress}).
	The number of \shorten{}s in which one vector has a smaller length is bounded by $4n$ (\Cref{lemma:shortenSmall}).

	It can happen that a \runState{} is stopped via the progress of a different run or two runs together have progress (in \jointShorten{}s or \jointMerge{}s).
	Every \merge{} or \jointMerge{} can stop at most $4$ other runs.
	A \jointShorten{} can stop one other run (in case the two runs form a joint run-pair and only one of them executes the \shorten{}).
	Thus, every \shorten{} stops at most one additional run.
	Every \jointShorten{} and \jointMerge{} also stops one additional run because two runs together have progress.
	Lastly, in case of \merge{} and \jointMerge{}s, at most $4$ other runs are stopped.
	Thus we count at most $8n$ runs for \merge{}s and \jointMerge{}s, $2 \cdot \frac{n \cdot \left(1+ \frac{1}{5}\right)}{0.019}$ for \shorten{}s and \jointShorten{}s with vectors of length at least $\frac{1}{2}$, and $8n$ \runState{}s for  \shorten{}s and \jointShorten{}s in which one vector has a length of at most $\frac{1}{2}$.
	The total number of \runState{}s is hence upper bounded by $16n + 2 \cdot \frac{\left(1+\frac{1}{5}\right)}{\shortenConstant} \leq 143n$.
\end{proof}

\numberOfRoundsLemma*

\begin{proof}
	We make use of a witness argument here.
	Consider an arbitrary robot $r_i$ with $\mathit{init}(r_i) = \mathit{true}$.
	$r_i$ tries every $7$ rounds to start two new \runState{}s.
	In case it does not start new \runState{}s, it either sees an other \runState{} or it is blocked by a \merge{} of an other \runState{} within the last $4$ rounds.
	Since every \runState{}s moves to the next robot in every round, $|N_i(t)| = 7$ and no run visits the same robot twice (\Cref{lemma:runStopLemma}), $r_i$ can at most twice be prevented from starting new \runState{}s by the same \runState{}.
	Thus, every $7$ rounds, $r_i$ either starts two new \runState{}s or is hindered by a \runState{}, however it can only be hindered by the same run twice.
	We say that in $7k$ rounds, for an integer $k$, $r_i$ is a witness of $k$ \runState{}s.

	It can however happen that $r_i$ executes a \merge{} in some round $t$.
	W.l.o.g., we assume that $r_i$ merges with $r_{i+1}$.
	Now, we have to consider three cases.
	Either $\mathit{init}(r_{i+1},t+1) = \mathit{true}$, $\mathit{init}(r_{i+2},t+1) = \mathit{true}$ or $\mathit{init}(r_{i+1},t+1) = \mathit{init}(r_{i+2},t+1) = \mathit{false}$.
	Assume that $\mathit{init}(r_{i+1},t+1) = \mathit{true}$.
	In the next iteration when $r_{i+1}$ tries to generate new \runState{}s it can happen that it is hindered by \runState{} that we have already seen at $r_i$.
	However, $7$ rounds later, this cannot be the case anymore because $r_i$ and $r_{i+1}$ have been direct neighbors.
	The same argument holds if $\mathit{init}(r_{i+2},t+1) = \mathit{true}$.

	What remains is a single case where neither $\mathit{init}(r_{i+1},t+1) = \mathit{true}$ nor $\mathit{init}(r_{i+2},t+1) = \mathit{true}$.
	This can only happen if $r_{i+2}$ has an \runInit{} in round $t$ and executes a \merge{} in the same round.
	Then, either $r_{i+3}$ or $r_{i+4}$ has an \runInit{} or none of them if $r_{i+4}$ has an \runInit{} in round $t$ and executes a \merge{}.
	However, at the end of such a sequence there either must be a robot that has an \runInit{} or all \runState{}s along the chain execute a \merge{} in round $t$.
	We either continue counting at the robot that has an \runInit{} at the end of such a sequence or (if all \runState{}s along the chain have executed a \merge{}) we continue counting at an arbitrary new \runInit{} that will be generated in the next round.
	In both cases, we do not count any \runState{} that we have already counted at $r_i$ again.

	All in all, at most every $14$ rounds either a new run is generated or we count a run while counting the same run at most twice.
	Thus, after at most $2 \cdot 14 \cdot n \cdot \left(16+2 \cdot \frac{\left(1+\frac{1}{5}\right)}{\shortenConstant}\right) \leq 4013n$ rounds we have counted enough runs that are necessary to gather all robots in a single point (\Cref{lemma:totalCountRuns}).
	After at most $n-1$ additional rounds, every run has ended and the configuration is gathered.
	As soon as only $5$ robots are remaining in the chain, all robots move a distance of $1$ towards the center of the smallest enclosing circle of all robots.
	After $5$ more rounds, the robots have gathered.
	The lemma follows.
\end{proof}
\subsection{Symmetric Case}

\isogonalToRegularStarLemma*

\begin{proof}
	Given that the entire configuration is an isogonal configuration, every robot executes a \schlaefli{}.
	All robots lie on the same circle in round $t$ and the \schlaefli{} ensures that all robots continue to stay on the same circle because only target points on the boundary of the circle are computed.
	In case the diameter of the circle has a size of at most $2$, the robots gather in round $t+1$ in the midpoint of the circle.
	Now assume that the diameter has a size of at least $2$.
	This implies that no pair of robots can be connected via a vector that describes the diameter of the circle.
	Thus, the circular arcs $L_{\alpha}$ and $L_{\beta}$ are unique.
	Consider now a robot $r_i$ and its neighbors $r_{i-1}$ and $r_{i+1}$.
	W.l.o.g.\ assume that $\|u_i(t)\| < \|u_{i+1}(t)\|$ and thus $L_{\alpha}$ connects $r_i$ and $r_{i-1}$ and $L_{\beta}$ connects $r_i$ and $r_{i+1}$.
	In the \schlaefli{}, the robots $r_i$ and $r_{i+1}$ move towards each other and the robots $r_i$ and $r_{i-1}$ move away from each other.
	Moreover, $r_i$ enlarges $L_{\alpha}$ by $r \cdot \left(\frac{\beta-\alpha}{4}\right)$ and reduces $L_{\beta}$ by the same value.
	Since $r_{i-1}$ enlarges $L_{\alpha}$ by the same distance and $r_{i+1}$ reduces $L_{\beta}$ by the same value, the new circular arcs can be computed as follows:

	\begin{enumerate}
		\item $L_{\alpha}(t+1) = L_{\alpha} + 2 \cdot r \cdot \left(\frac{\beta-\alpha}{4}\right) = r \cdot \alpha + r \cdot \left(\frac{\beta-\alpha}{2}\right) = r \cdot \left(\frac{\alpha + \beta}{2}\right)$
		\item $L_{\beta}(t+1) = L_{\beta} - 2 \cdot r \cdot \left(\frac{\beta-\alpha}{4}\right) = r \cdot \beta - r \cdot \left(\frac{\beta-\alpha}{2}\right) = r \cdot \left(\frac{\alpha + \beta}{2}\right)$
	\end{enumerate}

	Thus, $L_{\alpha}(t+1) = L_{\beta}(t+1)$. Since this holds for every robot, the configuration is a regular star configuration at time $t+1$.
\end{proof}

The next step is to prove that \bisector{}s executed in regular star configurations lead to a linear gathering time.
For this, we analyze the regular star configuration represented by the regular polygon $\{n/1\}$ with edge length $1$.
In $\{n/1\}$ and edge length $1$, the inner angles are of maximal size and the distances robots are allowed to move towards the center of the surrounding circle are minimal.
Thus, it is enough to prove a linear gathering time for $\{n/1\}$ with edge length $1$.
Afterwards, we can conclude a linear gathering time for any regular star configuration.
In the following, we remove for simplicity the assumption that a robot moves at most a distance of $\frac{1}{5}$ in a \bisector{}.
We multiply the resulting number of rounds with $5$ and get the same result.
For the proof, we introduce some additional notation.
Observe first that in a regular star configuration all angles $\alpha_{i}(t)$ have the same size, thus we simplify the notation to $\alpha(t)$ in this context.
Let $h(t) := \|p_i(t+1)-p_i(t)\|$ for any index $i$ be the distance of a robot to its target point.
This is again the same distance for every robot $r_i$.
In addition, define $d(t) := \|p_i(t)-p_{i+1}(t)\|$ for any index $i$.
The positions $p_i(t), p_{i+1}(t)$ and $p_i(t+1)$ form a triangle.
The angle between $h(t)$ and $d(t)$ has a size of $\frac{\alpha(t)}{2}$.
Let $y(t)$ denote the angle between $d(t)$ and the line segment connecting $p_{i+1}(t)$ and $p_i(t+1)$ and $\beta(t)$ be the angle between $h(t)$ and the line segment connecting $p_{i+1}(t)$ and $p_i(t+1)$.
See \Cref{figure:regularPolygonNotation} for a visualization of these definitions.

\begin{figure}[H]
	\centering
	\includegraphics[page=11, width=0.9\textwidth]{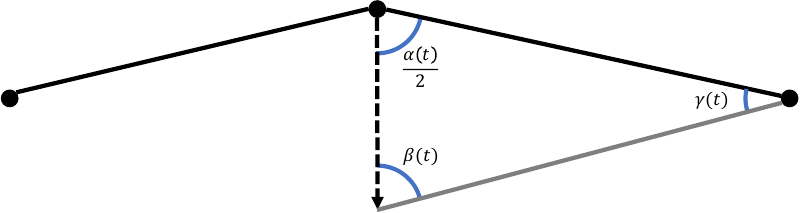}
	\caption{The notation for the analysis of regular star configurations.}
	\label{figure:regularPolygonNotation}
\end{figure}

Next, we state lemmas that analyze how the radius of the surrounding circle decreases.
Let $r(t)$ be the radius of the surrounding circle in round $t$ and $t_0$ the round in which the execution of the algorithm starts.

\begin{lemma} \label{lemma:regularPolygonRadius}
	$r(t+1) \leq r(t) - \sqrt{1- \frac{4 \pi^2 \cdot  r(t)^2}{n^2}}$.
\end{lemma}

\begin{proof}
	Since we are considering a regular polygon, $\alpha(t) = \frac{(n-2)\pi}{n}$ in every round $t$.
	By the law of sines, $h(t) = \frac{\sin (\gamma(t))}{\sin (\frac{\alpha(t)}{2})}$.
	Also, $\beta(t) = \arcsin (d(t) \cdot \sin (\frac{\alpha(t)}{2}) )$ and $\gamma(t) = \pi - \frac{\alpha(t)}{2} - \beta(t)$.
	Thus, $\sin(\gamma(t)) = \sin (\pi - \frac{\alpha(t)}{2} - \beta(t)) = \sin (\frac{\alpha(t)}{2} + \beta(t))$.
	Together with 	$\frac{r(t)}{r(t_0)} = \frac{d(t)}{d(t_0)}$ (intercept theorem), we obtain the following formula for $h(t)$:

	\begin{align*}
		h(t) = \frac{\sin \left(\gamma(t)\right)}{\sin \left(\frac{\alpha(t)}{2}\right)} & =  \frac{\sin \left(\frac{\alpha(t)}{2} + \beta(t) \right)}{\sin \left(\frac{\alpha(t)}{2}\right)}                                                                               \\
		& =  \frac{\sin \left(\frac{\alpha(t)}{2} + \arcsin \left(d(t) \cdot \sin \left( \frac{\alpha(t)}{2}\right)\right) \right)}{\sin \left(\frac{\alpha(t)}{2}\right)}                 \\
		& =  \frac{\sin \left(\frac{\alpha(t)}{2} + \arcsin \left( \frac{r(t)}{r(t_0)} \cdot \sin \left( \frac{\alpha(t)}{2}\right)\right) \right)}{\sin \left(\frac{\alpha(t)}{2}\right)} \\
		& =  \frac{\sin \left(\frac{(n-2)\pi}{2n} + \arcsin \left( \frac{r(t)}{r(t_0)} \cdot \sin \left( \frac{(n-2)\pi}{2n}\right)\right) \right)}{\sin \left(\frac{(n-2)\pi}{2n}\right)} \\
		& =  \frac{\sin \left(\frac{(n-2)\pi}{2n} + \arcsin \left( \frac{r(t)}{r(t_0)} \cdot \cos \left( \frac{\pi}{n}\right)\right) \right)}{\cos \left(\frac{\pi}{n}\right)}             \\
		& =  \frac{\sin \left(\frac{(n-2)\pi}{2n} + \arcsin \left( \frac{2\pi r(t)}{n} \cdot \cos \left( \frac{\pi}{n}\right)\right) \right)}{\cos \left(\frac{\pi}{n}\right)}             \\
		& =  \frac{\sin \left(\frac{\pi}{2} - \frac{\pi}{n} + \arcsin \left( \frac{2\pi r(t)}{n} \cdot \cos \left( \frac{\pi}{n}\right)\right) \right)}{\cos \left(\frac{\pi}{n}\right)}   \\
		& = \frac{\cos \left(\frac{\pi}{n} - \arcsin \left( \frac{2\pi r(t)}{n} \cdot \cos \left( \frac{\pi}{n}\right)\right) \right)}{\cos \left(\frac{\pi}{n}\right)}                    \\
	\end{align*}
	Next, we apply the trigonometric identities $\cos \left(x-y\right) = \sin \left(x\right) \cdot \sin \left(y\right) + \cos \left(x\right) \cdot \cos \left(y\right)$ and $\cos \left(\arcsin\left(x\right)\right) = \sqrt{1-x^2}$ and obtain:

	\begin{align*}
		h(t) & = \frac{\cos \left(\frac{\pi}{n} - \arcsin \left( \frac{2\pi r(t)}{n} \cdot \cos \left( \frac{\pi}{n}\right)\right) \right)}{\cos \left(\frac{\pi}{n}\right)} \\
		& = \sqrt{1- \frac{4\pi^2 \cdot r(t)^2}{n^2} \cdot \cos^2 \left(\frac{\pi}{n}\right)} + \frac{2\pi \cdot r(t)}{n} \cdot \sin \left(\frac{\pi}{n}\right)         \\
		& \geq \sqrt{1- \frac{4\pi^2 \cdot r(t)^2}{n^2} \cdot \cos^2 \left(\frac{\pi}{n}\right)}                                                                        \\
		& \geq \sqrt{1- \frac{4\pi^2 \cdot r(t)^2}{n^2}}                                                                                                                \\
	\end{align*}

	Finally, we can prove the lemma:

	\begin{align*}
		r(t+1)  =  r(t) - h(t)   \leq r(t) -  \sqrt{1- \frac{4\pi^2\cdot r(t)^2}{n^2}}
	\end{align*}

\end{proof}

Now, define by $\Delta r(t) = r(t_0)-r(t)$.
For $\Delta r(t+1)$ we can derive the formula stated by the following lemma.

\begin{lemma} \label{lemma:xt}
	For $\Delta r(t) \geq 1$:
	$\Delta r(t+1) \geq \Delta r(t) +  \frac{\sqrt{\Delta r(t)}}{\sqrt{n}}$
\end{lemma}

\begin{proof}
	First of all, observe $\Delta r(t+1)  = \Delta r(t) + h(t)$.

	\begin{align*}
		\Delta r(t+1) & \geq \Delta r(t) + \sqrt{1- \frac{4\pi^2 \cdot \left(r(t_0) - \Delta r(t) \right)^2}{n^2}}                                     \\
		& = \Delta r(t)  + \sqrt{1- \frac{4\pi^2 \cdot \left(\frac{n}{2\pi} - \Delta r(t) \right)^2}{n^2}}                               \\
		& = \Delta r(t) +  \sqrt{\frac{4\pi \cdot \Delta r(t)}{n} - \frac{4\pi^2 \cdot \Delta r(t)^2}{n^2}}                              \\
		& = \Delta r(t) +  \frac{2 \sqrt{\pi} \sqrt{\Delta r(t)}}{\sqrt{n}}\sqrt{1 - \frac{\pi \cdot \Delta r(t)}{n}}                    \\
		& \geq  \Delta r(t) +  \frac{2 \sqrt{\pi} \sqrt{\Delta r(t)}}{\sqrt{2}\sqrt{n}}                                                  \\
		& =  \Delta r(t) +  \frac{\sqrt{2} \sqrt{\pi} \sqrt{\Delta r(t)}}{\sqrt{n}}                                                      \\
		& \geq   \Delta r(t) +  \frac{\sqrt{\Delta r(t)}}{\sqrt{n}}                                                                      \\
	\end{align*}

\end{proof}

\begin{lemma} \label{lemma:initialXt}
	After $2n$ rounds, $\Delta r(t) \geq 1$.
\end{lemma}

\begin{proof}
	In the proof of  \Cref{lemma:regularPolygonRadius}, we identified $h(t) \geq \frac{2\pi \cdot r(t)}{n} \cdot \sin \left(\frac{\pi}{n}\right)$.
	Now assume that $r(t) \geq \frac{r(t_0)}{2} = \frac{n}{4\pi}$ (which holds for sufficiently large $n$ since $\Delta r(t)  < 1$).
	Then, $h(t) \geq \frac{2\pi \cdot n}{4 \pi} \cdot \sin \left(\frac{\pi}{n}\right) = \frac{1}{2} \cdot \sin \left(\frac{\pi}{n}\right)$.
	For $n \geq 2$ it now holds $h(t) \geq \frac{\pi}{4n}$.
	Thus, after $2n$ rounds it holds $\Delta r(t) \geq 1$.
\end{proof}

\begin{lemma} \label{lemma:regularStarPolygonsGatheringTime}
	After at most $6n$ rounds, $\Delta r(t) \geq \frac{n}{2 \pi}$ and thus $r(t) = 0$.
\end{lemma}

\begin{proof}
	We fix the first time step $t'$ such that $\Delta r(t') \geq 1$ (note that $\Delta r(t') \leq 2$ in this case since no robot moves more than distance $1$ per round).
	By \Cref{lemma:initialXt} this holds after at most $2n$ rounds.
	Furthermore, $\Delta r(t)$ doubles every $\sqrt{n} \cdot \sqrt{\Delta r(t)}$ rounds (\Cref{lemma:xt}).
	After at most $\log n$ doublings, it holds $\Delta r(t) \geq n$ and the robots are gathered.
	The first doubling requires $\sqrt{n} \cdot \sqrt{\Delta r(t')}$ rounds, the next doublings  $\sqrt{n} \cdot \sqrt{2\cdot \Delta r(t')}$, $\sqrt{n} \cdot \sqrt{4 \cdot \Delta r(t')}$, $\dots$ rounds.
	Thus, the number of rounds for $\log n$ doublings can be counted as follows:

	\begin{align*}
		\sqrt{\Delta r(t')} \cdot \sqrt{n}\sum_{k=1}^{\log n} \sqrt{2^k} = \sqrt{\Delta r(t')} \cdot \sqrt{n} \cdot \sqrt{2} \left(1+ \sqrt{2}\right) \left(\sqrt{n} -1\right) \leq 4n
	\end{align*}

	The total number of rounds can therefore be upper bounded by $4n + 2n = 6n$.

\end{proof}

\regularStarRoundLemma*

\begin{proof}
	In general, this bisectors of regular star configurations intersect in the center of $C$.
	Thus, with every \bisector{}, the distance of a robot to the center of $C$ decreases until finally all robots gather at the center.
	We prove the runtime for regular star configurations with Schl\"afli symbol $\{n/1\}$, which are also denoted as regular polygons.
	These polygons have inner angles $\alpha_{i}(t)$ of maximal size, for higher values of $d$ (referring to the Schl\"afli symbol $\{n/d\}$), the inner angles become smaller.
	Thus, for the regular polygon $\{n/1\}$, the distance a robot is allowed to move within a \bisector{} is minimal among all regular star polygons.
	\Cref{lemma:regularStarPolygonsGatheringTime} states that the regular star configuration $\{n/1\}$ with side length $1$ is gathered in $6n$ rounds.
	Since we did not consider the assumption that a robot moves a distance of at most $\frac{1}{5}$ per round, we multiply the result with $5$ and get an upper bound of $30n$ on the number of required rounds.
	Hence, all other regular star configurations can be gathered in linear time as the robot are allowed to move larger distances per round.

\end{proof}




%
%

\end{document}